\documentclass[a4paper,10pt]{article}

\usepackage[top=1.0in,right=1.in,left=1.in,bottom=1.75in]{geometry}
\usepackage{amssymb,amsbsy}
\usepackage{amsmath}
\usepackage{amsthm}
\usepackage{amsfonts}
\usepackage{cases}
\usepackage{caption}
\usepackage{subcaption}
\usepackage{floatrow}

\usepackage{graphics,graphicx}
\usepackage[T1]{fontenc}
\usepackage{color}
\usepackage{booktabs}
\usepackage[affil-it,auth-sc]{authblk}
\usepackage{stmaryrd}

\usepackage[]{jabbrv}

\usepackage[numbers,sort&compress]{natbib}

\newcommand{\beq}{\begin{equation}}
\newcommand{\eeq}{\end{equation}}
\newcommand{\lb}{\label}
\newcommand{\beqar}{\begin{eqnarray}}
\newcommand{\eeqar}{\end{eqnarray}}
\newcommand{\bit}{\begin{itemize}}
\newcommand{\eit}{\end{itemize}}
\newcommand{\benum}{\begin{enumerate}}
\newcommand{\eenum}{\end{enumerate}}
\newcommand{\barr}{\begin{array}}
\newcommand{\earr}{\end{array}}

\newtheorem{theorem}{Theorem}[section]

\def\XXint#1#2#3{{\setbox0=\hbox{$#1{#2#3}{\int}$}
   \vcenter{\hbox{$#2#3$}}\kern-.5\wd0}}

\def\b0{\mbox{\boldmath $0$}}

\def\be{\mbox{\boldmath $e$}}

\def\bn{\mbox{\boldmath $n$}}

\def\bp{\mbox{\boldmath $p$}}

\def\bu{\mbox{\boldmath $u$}}

\def\bx{\mbox{\boldmath $x$}}

\def\bA{\mbox{\boldmath $A$}}
\def\bB{\mbox{\boldmath $B$}}
\def\bC{\mbox{\boldmath $C$}}

\def\bE{\mbox{\boldmath $E$}}
\def\bF{\mbox{\boldmath $F$}}
\def\bG{\mbox{\boldmath $G$}}

\def\bI{\mbox{\boldmath $I$}}

\def\bK{\mbox{\boldmath $K$}}

\def\bR{\mbox{\boldmath $R$}}

\def\bU{\mbox{\boldmath $U$}}

\newcommand{\bsigma}{\mbox{\boldmath $\sigma$}}
\newcommand{\bSigma}{\mbox{\boldmath $\Sigma$}}

\def\f0{\ensuremath{\mathbb{O}}}

\newcommand{\mF}{\ensuremath{\mathcal{F}}}

\newcommand{\mI}{\ensuremath{\mathcal{I}}}

\newcommand{\mP}{\ensuremath{\mathcal{P}}}

\newcommand{\mS}{\ensuremath{\mathcal{S}}}
\newcommand{\mT}{\ensuremath{\mathcal{T}}}

\newcommand{\bmA}{\mbox{\boldmath $\mathcal{A}$}}
\newcommand{\bmB}{\mbox{\boldmath $\mathcal{B}$}}
\newcommand{\bmC}{\mbox{\boldmath $\mathcal{C}$}}

\newcommand{\tr}{\mathop{\mathrm{tr}}}

\def\Im{\mathop{\mathrm{Im}}}
\def\Re{\mathop{\mathrm{Re}}}

\newcommand{\sign}{\mathop{\mathrm{sign}}}

\newcommand{\ci}{\mathop{\mathrm{ci}}}
\newcommand{\si}{\mathop{\mathrm{si}}}

\newcommand{\Reals}{\ensuremath{\mathbb{R}}}

\newcommand{\ljump}{\llbracket}
\newcommand{\rjump}{\rrbracket}
\newcommand{\jump}[1]{\ljump #1 \rjump}
\newcommand{\av}[1]{\langle #1 \rangle}
\newfloatcommand{capbtabbox}{table}[][\FBwidth] 

\definecolor{corr2}{RGB}{1,1,255}   

\title{Boundary integral formulation for cracks at imperfect interfaces}

\author[1]{G. Mishuris}
\author[1,2]{A. Piccolroaz\footnote{Corresponding author: e-mail: roaz@ing.unitn.it; phone: +39\,0461\,282583.}}
\author[1]{A. Vellender}

\affil[1]{Institute of Mathematical and Physical Sciences, Aberystwyth University, Wales, U.K.}
\affil[2]{Department of Civil, Environmental and Mechanical Engineering, University of Trento, Italy}

\begin{document}

\maketitle

\begin{abstract} 
We consider an infinite bi-material plane containing a semi-infinite crack situated on a soft imperfect interface. The crack 
is loaded by a general asymmetrical system of forces distributed along the crack faces. On the basis of the weight function 
approach and the fundamental reciprocal identity, we derive the corresponding boundary integral formulation, relating physical 
quantities. The boundary integral equations derived in this paper in the imperfect interface setting show a weak singularity, 
in contrast to the perfect interface case, where the kernel is of the Cauchy type. We further present three alternative 
variants of the boundary integral equations which offer computationally favourable alternatives for certain sets of 
parameters.
\end{abstract}

\section{Introduction and motivation}
\label{sec1}

Integral equations play a crucial role in the fracture mechanics of dissimilar bodies. The use of singular integral equations in linear elasticity was first developed for solving two-dimensional problems by Mushkelishvili \cite{Mushkelishvili1953} and later extended to three-dimensional problems \cite{Kupradze1979,Mikhlin1980}. Typically, singular integral formulations have been derived through a Green's function approach, for which explicit expressions are required. While some such expressions for elastic isotropic and anisotropic materials have been found, integrals defining stresses and displacements are often numerically challenging to compute. Further, the Green's function approach requires the loadings applied on crack faces to be symmetric.

If the geometry permits, integral transform methods are a powerful tool to transform a problem to a system of integral equations with respect to the integral densities. For example, Fourier or Mellin transforms allow for the solution of numerous problems in layered and wedged domains respectively \cite{Craster1994,Sneddon1972,Uflyand1963}. As a result, Fredholm or weakly singular integral equations can be obtained, offering an efficient numerical tool for computing solutions for a range of problems in elasticity and fracture mechanics \cite{Antipov2000,Antipov2011,Malits2012}. Recently, this method was efficiently extended to more complex domains consisting of an arbitrary number of different layers and wedges \cite{Mishuris1997} and extended for the case of imperfect interfaces between the subdomains \cite{Mishuris1997a,Mishuris1997b}. This allowed for the analysis of the asymptotic behaviour of mechanical fields in a neighbourhood of the crack tip situated at various imperfect interfaces \cite{Mishuris2001,Mishuris2001b}. The respective problems 
were transformed to systems of singular integral equations with so-called fixed point singularities, whose symbols and indices were analysed. This type of integral equation was analysed in Duduchava \cite{Duduchava1979}, based on the abstract theory of linear singular operators \cite{Gohberg1958,Krein1958}.

In essence, the integral transform methods dealing with the densities of the respective integral not with the physical measures rely on the specific problem (boundary conditions) defined. A more general approach is based on the integral identities method. It is not restricted to any specific geometry and deals with physical values (displacement discontinuities and tractions) along the cracks and interphases. This method has an additional advantage as the identities can be used to solve more complex multiphysics problems relating various fields. For instance, boundary element methods (BEM) enable effective treatment of problems involving nonlinearity that can result from a number of physical phenomena such as temperature-dependent material properties, chemical reactions or nonlinear sources \cite{Wrobel2002}.

Recently, an approach based on Fourier transforms, Betti's reciprocal theorem and weight functions (singular non-trivial solutions of the homogeneous traction free problem, see \cite{Bueckner1987,Willis1995}) has been employed to derive integral equations which relate applied loadings on crack faces to the resulting crack opening displacements. This approach, relating physical quantities defined along boundaries, allows for the use of asymmetric loadings and has been used to derive identities for dissimilar isotropic \cite{Piccolroaz2013} and anisotropic \cite{Morini2013} materials. 

The aim of the present paper is to derive analogous integral identities for the case of semi-infinite interfacial cracks in isotropic bimaterials that are joined by a soft imperfect interface. Soft imperfect interfaces model a thin layer of  adhesive material between two larger bodies. Typically such a thin layer is replaced in problem formulations by a condition that the jump in displacement across the interface is proportional to the tractions along the interface; this approach is justified for example in \cite{Antipov2001,Mishuris2001}. {In fact, such types of interface appear earlier in the literature without a rigorous derivation, for example as the Winkler foundation \cite{Kavlakan1980,Kavlakan1981}, in composites \cite{Erdogan1997}, and in problems relating to mining \cite{Linkov1974,Malits1982}. The validity of transmission conditions in the case of a singular point has been verified \cite{Mishuris2001a} and also verification by finite element analysis in the case of linear and non-linear interphases 
has been conducted \cite{Mishuris2005,Mishuris2007}.} 

Recently, weight functions have been derived for strip and plane geometries containing cracks and imperfect interfaces \cite{Vellender2011,Vellender2013}; these weight functions allow for the evaluation of constants that may be used in fracture criteria and have been used to obtain other information about a structure's behaviour \cite{Vellender2012}.

\subsection{Problem formulation}

\begin{figure}[!htcb]
\centering
\includegraphics[width=13cm]{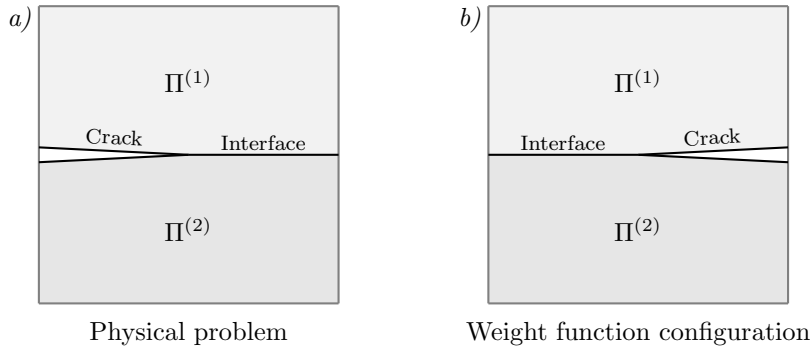}
\caption{\footnotesize {The geometries of the physical problem ($a$) described in Subsection 1.1, and the weight function problem configuration ($b$) described in Section 2.}}
\label{fig01}
\end{figure}

We consider an infinite bi-material plane with an imperfect interface positioned along the $x$-axis {as shown in Figure \ref{fig01}a}. A semi-infinite crack is placed at the interface occupying the line
$\Gamma = \{(x_1,x_2): x_1 < 0, x_2 = 0\}$. We refer to the half-planes above and below the interface as $\Pi^{(1)}$ and $\Pi^{(2)}$, respectively. The material
occupying $\Pi^{(j)}$ has shear modulus $\mu_j$ and Poisson's ratio $\nu_j$ for $j = 1,2$.

We introduce the imperfect interface conditions ahead of the crack ($x_1 > 0$):
\begin{equation}
\label{int1}
\bsigma_2(x_1,0^+) = \bsigma_2(x_1,0^-),
\end{equation}
\begin{equation}
\label{int2}
\bu(x_1,0^+) - \bu(x_1,0^-) = \bK \bsigma_2(x_1,0^+),
\end{equation}
where $\bsigma_2 = (\sigma_{21},\sigma_{22},\sigma_{23})^T$ denotes the traction vector, $\bu = (u_1,u_2,u_3)^T$ the displacement field and the matrix $\bK$ describes 
the extent of imperfection of the interface. Note that in case of isotropic interface layer $\bK$ is a diagonal matrix, whereas in case of anisotropic interface layer 
$\bK$ is a symmetric matrix as shown in \cite{Antipov2001}, with the following structure
\begin{equation}
\bK =
\begin{pmatrix}
K_{11} & K_{12} & 0      \\
K_{12} & K_{22} & 0      \\
0      & 0      & K_{33}
\end{pmatrix}.
\end{equation}
{We will consider in this paper the cases for which the components of $\bK$ are constants, and $\bK$ is positive definite. For such cases, the stresses along the interface remain bounded at the crack tip \cite{Antipov2001,Mishuris2001}. Otherwise, the interface imperfection matrix $\bK$ may depend on $\bx$, which significantly alters the behaviour of the physical fields near the crack tip \cite{Mishuris1998,Mishuris1999,Mishuris2001,Mishuris2001b}.}

The crack faces are loaded by a system of, not necessarily symmetrical, distributed forces ($x_1 < 0$)
\begin{equation}
\bsigma_2(x_1,0^+) = \bp_+(x_1), \quad \bsigma_2(x_1,0^-) = \bp_-(x_1).
\end{equation}
It is convenient to introduce the symmetrical and skew-symmetrical parts of the loading as follows
\begin{equation}
\av{\bp}(x_1) = \frac{1}{2}(\bp_+(x_1) + \bp_-(x_1)), \quad \jump{\bp}(x_1) = \bp_+(x_1) - \bp_-(x_1),
\end{equation}
where we used standard notations to denote the average, $\av{f}$, and the jump, $\jump{f}$, of a function $f$ across the crack/interface line, $x_2 = 0$,
\begin{equation}
\av{f}(x_1) = \frac{1}{2} [f(x_1,0^+) + f(x_1,0^-)], \quad \jump{f}(x_1) = f(x_1,0^+) - f(x_1,0^-).
\end{equation}
In this paper we adopt the approach based on weight functions in the imperfect interface setting found by Vellender et al. \cite{Vellender2013}. In particular, the weight functions $\bU$ 
and $\bSigma$ must satisfy the following transmission conditions
\begin{equation}
\label{int1wf}
\bSigma_2(x_1,0^+) = \bSigma_2(x_1,0^-),{\quad x_1\in\mathbb{R},}
\end{equation}
\begin{equation}
\label{int2wf}
\bU(x_1,0^+) - \bU(x_1,0^-) = \bK^* \bSigma_2(x_1,0^+),{\quad x_1<0,}
\end{equation}
where
\begin{equation}
\bK^* =
\begin{pmatrix}
K_{11} & -K_{12} & 0 \\
-K_{12} & K_{22} & 0 \\
0 & 0 & K_{33}
\end{pmatrix}.
\end{equation}

\section{The Betti formula and weight functions in the imperfect interface setting}
\label{sec2}

In this section, we extend the Betti formula to the case of general asymmetrical loading applied at the crack surfaces. The Betti formula is used in order to relate
the physical solution to the weight function, which is a special singular solution to the homogeneous problem (traction-free crack faces) (see \cite{Willis1995} and 
\cite{Piccolroaz2007}).

In the absence of body forces, the Betti formula takes the form
\begin{equation}
\int_{\partial\Omega} \Big\{ \bsigma^{(a)} \bn \cdot \bu^{(b)} - \bsigma^{(b)} \bn \cdot \bu^{(a)} \Big\} ds = 0,
\end{equation}
where $\partial\Omega$ is any surface enclosing a region $\Omega$ within which both displacement fields $\bu^{(a)}$ and $\bu^{(b)}$ satisfy the
equations of equilibrium, with corresponding stress states $\bsigma^{(a)}$ and $\bsigma^{(b)}$, and $\bn$ denotes the outward normal to
$\partial\Omega$.

Applying the Betti formula to a semicircular domain in the upper half-plane $\Pi^{(1)}$, whose straight boundary is $x_2 = 0^+$ and whose radius
$R$ will be allowed to tend to infinity, we obtain, in the limit $R \to \infty$,
\begin{equation}
\label{betti}
\int\limits_{(x_2=0^+)}
\Big\{ \bsigma_2^{(a)}(x_1,0^+) \cdot \bu^{(b)}(x_1,0^+) - \bsigma_2^{(b)}(x_1,0^+) \cdot \bu^{(a)}(x_1,0^+) \Big\} dx_1 = 0,
\end{equation}
provided that the fields $\bu^{(a)}$ and $\bu^{(b)}$ decay suitably fast at infinity. The notation $\bsigma_2$ is used to denote the traction
vector acting on the plane $x_2 = 0$: $\bsigma_2 = \bsigma \be_2$.

We can assume that $\bu^{(a)}$ represents the physical field associated with the crack loaded at its surface, whereas $\bu^{(b)}$ represents a
non-trivial solution of the homogeneous problem, the so-called weight function, defined as follows
\begin{equation}
\lb{trans}
\bu^{(b)}(x_1,x_2) = \bR \bU(-x_1,x_2),
\end{equation}
where $\bR$ is a rotation matrix
\begin{equation}
\bR =
\begin{pmatrix}
-1 & 0 &  0 \\
 0 & 1 &  0 \\
 0 & 0 & -1
\end{pmatrix}.
\end{equation}
Note that the transformation (\ref{trans}) corresponds to introducing a change of coordinates in the solution $\bu^{(b)}$, namely a
rotation about the $x_2$-axis through an angle $\pi$ {(see Figure \ref{fig01}b)}. It is straightforward to verify that the weight function $\bU$ satisfies the equations of
equilibrium, but in a different domain where the crack is placed along the semi-plane $x_2 = 0$, $x_1 >0$.
The notation $\bSigma$ will be used for components of stress corresponding to the displacement field $\bU$,
\begin{equation}
\bsigma^{(b)}(x_1,x_2) = \bR \bSigma(-x_1,x_2) \bR.
\end{equation}
Replacing $\bu^{(b)}(x_1, x_2)$ with $\bu^{(b)}(x_1 - x_1', x_2)$, we obtain
\begin{equation}
\label{eq20}
\int\limits_{(x_2=0^+)}
\Big\{ \bR \bU(x_1'-x_1,0^+) \cdot \bsigma_2(x_1,0^+) - \bR \bSigma_2(x_1'-x_1,0^+) \cdot \bu(x_1,0^+) \Big\} dx_1 = 0.
\end{equation}
A similar equation can be derived by applying the Betti formula to a semicircular domain in the lower half-plane $\Pi^{(2)}$,
\begin{equation}
\label{eq21}
\int\limits_{(x_2=0^-)}
\Big\{ \bR \bU(x_1'-x_1,0^-) \cdot \bsigma_2(x_1,0^-) - \bR \bSigma_2(x_1'-x_1,0^-) \cdot \bu(x_1,0^-) \Big\} dx_1 = 0.
\end{equation}
Subtracting (\ref{eq21}) from (\ref{eq20}), we obtain
\begin{equation}
\label{recid}
\int\limits_{(x_2=0)}
\Big\{ \bR \jump{\bU}(x_1'-x_1) \cdot \langle \bsigma_2 \rangle(x_1) +
\bR \langle \bU \rangle(x_1'-x_1) \cdot \jump{\bsigma_2}(x_1) \\ 
- \bR \langle \bSigma_2 \rangle(x_1'-x_1) \cdot \jump{\bu}(x_1) \Big\} dx_1 = 0.
\end{equation}
Let us introduce the notations
\begin{equation}
\label{supers}
f^{(+)}(x_1) = f(x_1) H(x_1), \quad f^{(-)}(x_1) = f(x_1) H(-x_1),
\end{equation}
where $H$ denotes the Heaviside function, so that
\[
f(x_1) = f^{(+)}(x_1) + f^{(-)}(x_1).
\]
This allows the splitting of the physical stress terms into two parts, as follows
\begin{equation}
\av{\bsigma}(x_1) = \av{\bsigma}^{(+)}(x_1) + \av{\bsigma}^{(-)}(x_1) = \av{\bsigma}^{(+)}(x_1) + \av{\bp}(x_1),
\end{equation}
\begin{equation}
\jump{\bsigma}(x_1) = \jump{\bsigma}^{(+)}(x_1) + \jump{\bsigma}^{(-)}(x_1) = \jump{\bp}(x_1),
\end{equation}
given that $\jump{\bsigma}^{(+)}(x_1) = 0$ from the interface condition \eqref{int1}.

The reciprocity identity (\ref{recid}) becomes
\begin{multline}
\label{recid2}
\int_{-\infty}^{\infty} \Big\{
\bR \jump{\bU}(x_1'-x_1) \cdot \av{\bsigma_2}^{(+)}(x_1) - \bR \av{\bSigma_2}(x_1'-x_1) \cdot \jump{\bu}(x_1)
\Big\} dx_1 = \\
- \int_{-\infty}^{\infty} \Big\{
\bR \jump{\bU}(x_1'-x_1) \cdot \av{\bp}(x_1) + \bR \av{\bU}(x_1'-x_1) \cdot \jump{\bp}(x_1)
\Big\} dx_1.
\end{multline}
The identity (\ref{recid2}) can be written in an equivalent form using the convolution with respect to $x_1$,
\begin{equation}
\label{recid3}
\bR \jump{\bU} * \av{\bsigma_2}^{(+)} - \bR \av{\bSigma_2}^{(-)} * \jump{\bu} =
- \bR \jump{\bU} * \av{\bp} - \bR \av{\bU} * \jump{\bp}.
\end{equation}
Note that we have not specified so far the exact nature of weight functions $\bU$ and $\bSigma$ used in this analysis, that is the boundary conditions and the 
character of the singularity near the crack tip and at infinity. Thus the identity (\ref{recid3}) is a universal one, valid for a large class of weight functions. 
In the case of perfect interface, the corresponding analysis has been conducted in Piccolroaz and Mishuris \cite{Piccolroaz2013} and Morini et al. \cite{Morini2013}.

In the sequel of the paper, we derive boundary integral equations in their most general form for two-dimensional deformations in the imperfect interface setting. 
We begin with the scalar case (Mode III) in the next section, in order to describe in details the procedure. In Sec. \ref{secmode12}, we analysis the vectorial case 
(Mode I and II) avoiding technical details.

\section{Boundary integral equations for imperfect interface. Mode III}
\label{secmode3}

\subsection{Evaluation of the boundary integral equations.}

In the case of antiplane deformation, the Betti formula (\ref{recid3}) relating the physical field $u_3 = u$, $\sigma_{23} = \sigma$ with
the weight function $U_3 = U$, $\Sigma_{23} = \Sigma$ reduces to the scalar equation
\begin{equation}
\label{m3id1}
\jump{U} * \av{\sigma}^{(+)} - \av{\Sigma}^{(-)} * \jump{u} = -\jump{U} * \av{p} - \av{U} * \jump{p},
\end{equation}
where $p = p_3$ is the antiplane loading applied on the crack faces.
Let us introduce the Fourier transform with respect to the variable $x_1$ as follows
\begin{equation}
\tilde{f}(\xi) = \mF_{\xi}[f(x_1)] = \int_{-\infty}^{\infty} f(x_1) e^{i\xi x_1} dx_1, \quad
f(x_1) = \mF^{-1}_{x_1}[\tilde{f}(\xi)] = \frac{1}{2\pi} \int_{-\infty}^{\infty} \tilde{f}(\xi) e^{-i\xi x_1} d\xi.
\end{equation}
Taking Fourier transforms in $x_1$ yields
\begin{equation}
\label{m3id2}
\jump{\tilde{U}}(\xi) \av{\tilde{\sigma}}^+(\xi) - \av{\tilde{\Sigma}}^-(\xi) \jump{\tilde{u}}(\xi) =
-\jump{\tilde{U}}(\xi) \av{\tilde{p}}(\xi) - \av{\tilde{U}}(\xi) \jump{\tilde{p}}(\xi), \quad \xi \in \Reals.
\end{equation}
where the superscripts $^+$ and $^-$ denote functions analytic in the upper and lower half-planes, respectively.

We now split $\jump{\tilde{U}}$ into the sum of $\jump{\tilde{U}}^\pm$ and similarly split $\jump{\tilde{u}}$ into the sum of $\jump{\tilde{u}}^\pm$
\begin{multline}
\label{m3id3}
\jump{\tilde{U}}^{+}(\xi) \av{\tilde{\sigma}}^+(\xi) + \jump{\tilde{U}}^{-}(\xi) \av{\tilde{\sigma}}^+(\xi) - \av{\tilde{\Sigma}}(\xi) \jump{\tilde{u}}^{+}(\xi)
- \av{\tilde{\Sigma}}(\xi) \jump{\tilde{u}}^{-}(\xi) = \\
-\jump{\tilde{U}}(\xi) \av{\tilde{p}}(\xi) - \av{\tilde{U}}(\xi) \jump{\tilde{p}}(\xi), \quad \xi \in \Reals.
\end{multline}
We now make use of the transmission conditions which state that
\begin{equation}
\jump{\tilde{U}}^-(\xi) = \kappa \av{\tilde{\Sigma}}^-(\xi), \quad \jump{\tilde{u}}^+(\xi) = \kappa \av{\tilde{\sigma}}^+(\xi),
\end{equation}
where $\kappa = K_{33}$. This causes the second and third terms in the left hand side of \eqref{m3id3} to cancel, leaving
\begin{equation}
\label{m3id4}
\jump{\tilde{U}}^+(\xi) \av{\tilde{\sigma}}^+(\xi) - \av{\tilde{\Sigma}}^-(\xi) \jump{\tilde{u}}^-(\xi) =
-\jump{\tilde{U}}(\xi) \langle \tilde{p} \rangle(\xi) - \langle \tilde{U} \rangle(\xi) \jump{\tilde{p}}(\xi), \quad \xi \in \Reals.
\end{equation}
We can now divide both sides of \eqref{m3id2} by $\jump{\tilde{U}}^+$ to obtain
\begin{equation}
\label{m3id5}
\av{\tilde{\sigma}}^+(\xi) - B(\xi) \jump{\tilde{u}}^-(\xi) = - [1 + \kappa B(\xi)] \langle \tilde{p} \rangle(\xi) - A(\xi) \jump{\tilde{p}}(\xi),
\end{equation}
where the factors in front of unknown functions are given by
\begin{equation}
A(\xi) = \frac{\av{\tilde{U}}(\xi)}{\jump{\tilde{U}}^+(\xi)}, \quad
B(\xi) = \frac{\av{\tilde{\Sigma}}^-(\xi)}{\jump{\tilde{U}}^+(\xi)}.
\end{equation}
They can be computed from the general relationships, which hold for the symmetric and skew-symmetric weight functions (Vellender et al., 2013)
\begin{equation}
\jump{\tilde{U}}^+(\xi) = -\kappa \left( 1 + \frac{\xi_0}{|\xi|} \right) \av{\tilde{\Sigma}}^-(\xi), \quad
\av{\tilde{U}}(\xi) = -\frac{\mu_*}{2} \jump{\tilde{U}}(\xi),
\end{equation}
where
\begin{equation}
\label{xi0}
\xi_0 = \frac{\mu_1 + \mu_2}{\kappa\mu_1\mu_2}, \quad \mu_* = \frac{\mu_1 - \mu_2}{\mu_1 + \mu_2}.
\end{equation}
%
Thus, we can easily obtain
\begin{equation}
A(\xi) = -\frac{\mu_*}{2}[1 + \kappa B(\xi)], \quad B(\xi) = - \frac{|\xi|}{\kappa|\xi| + \kappa\xi_0}.
\end{equation}
If we apply the inverse Fourier transform to \eqref{m3id5}, we obtain for the two opposite cases $x_1 < 0$ and $x_1 > 0$ the following relationships:
\begin{equation}
\label{m3id5a1}
\mF^{-1}_{x_1 < 0}\Big[ [1 + \kappa B(\xi)] \langle \tilde{p} \rangle(\xi) \Big] - \frac{\mu_*}{2} \mF^{-1}_{x_1 < 0}\Big[ [1 + \kappa B(\xi)] \jump{\tilde{p}}(\xi) \Big] =
\mF^{-1}_{x_1 < 0}\Big[ B(\xi) \jump{\tilde{u}}^-(\xi) \Big],
\end{equation}
\begin{equation}
\label{m3id5a2}
\av{\sigma}^{(+)}(x_1) = \mF^{-1}_{x_1 > 0}\Big[ B(\xi) \jump{\tilde{u}}^-(\xi) \Big] - \mF^{-1}_{x_1 > 0}\Big[ [1 + \kappa B(\xi)] \langle \tilde{p} \rangle(\xi) \Big]  
+ \frac{\mu_*}{2} \mF^{-1}_{x_1 > 0}\Big[ [1 + \kappa B(\xi)] \jump{\tilde{p}}(\xi) \Big],
\end{equation}
Note that the term $\av{\tilde{\sigma}}^+$ in \eqref{m3id5} cancels from \eqref{m3id3} because it is a ``$+$'' function.

Finally, we use the following inversion formulae
\begin{align}
\label{not}
\mF^{-1}\Big[ B(\xi) \tilde{f}(\xi) \Big] &= \frac{1}{\pi \kappa} (S_{\xi_0} * f')(x_1), \\[3mm]
\label{not2}
\mF^{-1}\Big[ [1 + \kappa B(\xi)] \tilde{f}(\xi) \Big] &= -\frac{\xi_0}{\pi} (T_{\xi_0} * f)(x_1),
\end{align}
where
\begin{align}
\label{S}
S_{\xi_0}(x) &= \sign(x)\si(\xi_0 |x|)\cos(\xi_0 |x|) - \sign(x)\ci(\xi_0 |x|)\sin(\xi_0 |x|), \\[3mm]
\label{T}
T_{\xi_0}(x) &= \si(\xi_0 |x|)\sin(\xi_0 |x|) + \ci(\xi_0 |x|)\cos(\xi_0 |x|),
\end{align}
and $\si$ and $\ci$ are the sine and cosine integral functions respectively, defined as
\begin{equation}
\si(x) = -\int_x^\infty \frac{\sin t}{t} dt = -\frac{\pi}{2} + \int_0^x \frac{\sin t}{t} dt,
\end{equation}
\begin{equation}
\ci(x) = -\int_x^\infty \frac{\cos t}{t} dt = \gamma + \ln x + \int_0^x \frac{\cos t - 1}{t} dt,
\end{equation}
in which $\gamma$ is the Euler-Mascheroni constant. The function $S_{\xi_0}$ has behaviour near zero and infinity described by 
\begin{equation}
\label{Sbehave}
S_{\xi_0}(x)\sim \mathrm{sign}(x)\left[-\frac{\pi}{2}+\xi_0(1-\gamma)|x|-\xi_0|x|\ln(\xi_0|x|)+O(|x|^2)\right],\quad x \to 0,
\end{equation}
\begin{equation}
S_{\xi_0}(x)\sim \mathrm{sign}(x)\left[-\frac{1}{\xi_0|x|}+O\left(\frac{1}{|x|^3}\right)\right],\quad x \to \pm \infty,
\end{equation}
while $T_{\xi_0}$ acts as 
\begin{equation}
\label{Tbehave}
T_{\xi_0}(x)\sim \ln(\xi_0|x|)+\gamma-\frac{\pi\xi_0}{2}|x|+O(|x|^2),\quad x \to 0,
\end{equation}
\begin{equation}
T_{\xi_0}(x)\sim -\frac{1}{\xi_0^2|x|^2}+O\left(\frac{1}{|x|^3}\right),\quad x \to \pm \infty.
\end{equation}
To simplify notations, we introduce the integral operators $\mS_{\xi_0}$, $\mT_{\xi_0}$ and the orthogonal projectors $\mP_\pm$ ($\mP_+ + \mP_- = \mI$) acting on 
the real axis:
\begin{equation}
\label{in1}
\mS_{\xi_0} \varphi(x) = 
(S_{\xi_0} * \varphi)(x) = \int_{-\infty}^{\infty} S_{\xi_0}(x - t) \varphi(t) dt,
\end{equation}
\begin{equation}
\label{in2}
\mT_{\xi_0} \varphi(x) = 
(T_{\xi_0} * \varphi)(x) = \int_{-\infty}^{\infty} T_{\xi_0}(x - t) \varphi(t) dt,
\end{equation}
\begin{equation}
\mP_\pm \varphi(x) =
\left\{
\begin{array}{ll}
\varphi(x), & \pm x \ge 0, \\[3mm]
0, & \text{otherwise}.
\end{array}
\right.
\end{equation}
The integral identities \eqref{m3id5a1} and \eqref{m3id5a2} become
\begin{equation}
\label{res1}
-\frac{\xi_0}{\pi}\mT_{\xi_0}^{(s)} \langle p \rangle(x_1) + \frac{\mu_*\xi_0}{2\pi}  \mT_{\xi_0}^{(s)} \jump{p}(x_1) = 
\frac{1}{\pi \kappa}\mS_{\xi_0}^{(s)} \frac{\partial \jump{u}^{(-)}}{\partial x_1} - \frac{1}{\pi \kappa} \jump{u}^{(-)}(0^-)S_{\xi_0}(x_1), \quad x_1 < 0,
\end{equation}
\begin{equation}
\label{res2}
\av{\sigma}^{(+)}(x_1) =
\frac{1}{\pi \kappa}\mS_{\xi_0}^{(c)} \frac{\partial \jump{u}^{(-)}}{\partial x_1} - \frac{1}{\pi \kappa} \jump{u}^{(-)}(0^-)S_{\xi_0}(x_1)  
+ \frac{\xi_0}{\pi}\mT_{\xi_0}^{(c)} \langle p \rangle(x_1) 
- \frac{\mu_*\xi_0}{2\pi} \mT_{\xi_0}^{(c)} \jump{p}(x_1), \quad x_1 > 0,
\end{equation}
where 
\begin{equation}
\label{ops}
\mS_{\xi_0}^{(s)} = \mP_- \mS_{\xi_0} \mP_- \qquad \text{and} \qquad \mT_{\xi_0}^{(s)} = \mP_- \mT_{\xi_0} \mP_-
\end{equation}
are singular operators. The terms involving $S_{\xi_0}(x_1)$ result from the fact that $\jump{u}^{(-)}(x_1)$ is discontinuous at $x_1=0$. It follows from (\ref{Sbehave}) and (\ref{Tbehave}) that  $\mT_{\xi_0}^{(s)}$ has a weak logarithmic singularity at zero, while $\mS_{\xi_0}^{(s)}$ has a Cauchy type singularity at infinity. The operators
\begin{equation}
\label{opc}
\mS_{\xi_0}^{(c)} = \mP_+ \mS_{\xi_0} \mP_- \qquad \text{and} \qquad \mT_{\xi_0}^{(c)} = \mP_+ \mT_{\xi_0} \mP_-
\end{equation}
{are respectively fixed-point singular and weakly singular in the appropriate functional spaces; for details we refer the reader to \cite{Gakhov1978}, \cite{Gohberg1958}, and \cite{Krein1958}.}

The two equations (\ref{res1}) and (\ref{res2}) form the system of integral identities for the antiplane deformation. The first equation (\ref{res1}) provides
the integral relationship between the applied loading $\langle p \rangle$, $\jump{p}$ and the resulting crack opening $\jump{u}^{(-)}$. This is
a singular integral equation and it is, generally speaking, invertible. However, the inverse operator depends on the properties of the solution.
The second equation (\ref{res2}) can be considered as an additional equation which allows to define the proper behaviour of the solution $\jump{u}^{(-)}$
and also, after the first equation has been inverted, to evaluate the traction ahead of the crack tip $\av{\sigma}^+$. {The reason for this is that
the operators in the right-hand side of (\ref{res2}) are fixed-point singular and weakly singular and thus are not invertible.}

\subsection{Alternative integral formulae}
\label{section:altmode3}

The integral identities \eqref{res1} and \eqref{res2}, derived in the previous section for antiplane deformations, appear to be natural representations in the imperfect 
interface setting. Moreover, if the extent of imperfection becomes infinitesimally small, $\kappa \to 0$, it is possible to prove that these formulae transform to 
the integral identities obtained in Piccolroaz and Mishuris \cite{Piccolroaz2013} for the perfect interface case (see Appendix).

However, in the case of imperfect interface it is possible to write the identities in alternative equivalent forms. Depending on the properties of the loading applied 
on the crack faces and/or the nature of the sought solution, one can then choose the more appropriate formulae.

Indeed, the following identity follows immediately from (\ref{not}), (\ref{not2}), (\ref{in1}) and (\ref{in2}):
\begin{equation}
\label{identi}
-\frac{\xi_0}{\pi}\mT_{\xi_0}\varphi = \mI\varphi + \frac{1}{\pi}\mS_{\xi_0}\varphi',
\end{equation}
if one assumes that both sides of the expression make sense for the solution in question. 

Exploiting information on the smoothness of the given functions $p_\pm$ and the identity (\ref{identi}) the integral identities (\ref{res1}) and (\ref{res2}) can be 
rewritten in the alternative form:
\begin{equation}
\label{res1b}
\frac{1}{\pi} \mS_{\xi_0}^{(s)} \frac{\partial \av{p}}{\partial x_1} + \av{p} - \frac{\mu_*}{2\pi} \mS_{\xi_0}^{(s)} \frac{\partial \jump{p}}{\partial x_1} 
- \frac{\mu_*}{2} \jump{p} = -\frac{\xi_0}{\pi\kappa} \mT_{\xi_0}^{(s)} \jump{u}^{(-)} - \frac{1}{\kappa} \jump{u}^{(-)}, \quad x_1 < 0,
\end{equation}
\begin{equation}
\label{res2b}
\av{\sigma}^{(+)} =
-\frac{\xi_0}{\pi\kappa} \mT_{\xi_0}^{(c)} \jump{u}^{(-)} - \frac{1}{\pi} \mS_{\xi_0}^{(c)} \frac{\partial \av{p}}{\partial x_1} 
+ \frac{\mu_*}{2\pi} \mS_{\xi_0}^{(c)} \frac{\partial \jump{p}}{\partial x_1}, \quad x_1 > 0.
\end{equation}
Here we have taken into account that $\av{p}$, $\jump{p}$ and $\jump{u}^{(-)}$ are equal to zero for $x_1 > 0$.

This is only one of the possible representations available when combining identity (\ref{identi}) with the integral equations (\ref{res1}) and (\ref{res2}). For 
example, we present also representations using only the operator $\mT_{\xi_0}$:
\begin{equation}
\label{res3b}
-\frac{\xi_0}{\pi} \mT_{\xi_0}^{(s)} \av{p} + \frac{\mu_*\xi_0}{2\pi} \mT_{\xi_0}^{(s)} \jump{p}  = 
-\frac{\xi_0}{\pi\kappa} \mT_{\xi_0}^{(s)} \jump{u}^{(-)} - \frac{1}{\kappa} \jump{u}^{(-)}, \quad x_1 < 0,
\end{equation}
\begin{equation}
\label{res4b}
\av{\sigma}^{(+)} = 
-\frac{\xi_0}{\pi\kappa} \mT_{\xi_0}^{(c)} \jump{u}^{(-)} + \frac{\xi_0}{\pi} \mT_{\xi_0}^{(c)} \av{p} 
- \frac{\mu_*\xi_0}{2\pi} \mT_{\xi_0}^{(c)} \jump{p}, \quad x_1 > 0,
\end{equation}
and representations using only the operator $\mS_{\xi_0}$:
\begin{equation}
\label{res5b}
\frac{1}{\pi} \mS_{\xi_0}^{(s)} \frac{\partial \av{p}}{\partial x_1} + \av{p} - \frac{\mu_*}{2\pi} \mS_{\xi_0}^{(s)} \frac{\partial \jump{p}}{\partial x_1} 
- \frac{\mu_*}{2} \jump{p} = \frac{1}{\pi\kappa} \mS_{\xi_0}^{(s)} \frac{\partial \jump{u}^{(-)}}{\partial x_1}, \quad x_1 < 0,
\end{equation}
\begin{equation}
\label{res6b}
\av{\sigma}^{(+)} = 
\frac{1}{\pi\kappa} \mS_{\xi_0}^{(c)} \frac{\partial \jump{u}^{(-)}}{\partial x_1} - \frac{1}{\pi} \mS_{\xi_0}^{(c)} \frac{\partial \av{p}}{\partial x_1} 
+ \frac{\mu_*}{2\pi} \mS_{\xi_0}^{(c)} \frac{\partial \jump{p}}{\partial x_1}, \quad x_1 > 0.
\end{equation}

\subsection{{Advantages of alternative formulations}}

The integral equations presented in equations (\ref{res1}), (\ref{res1b}), (\ref{res3b}) and (\ref{res5b}) offer four equivalent relationships between the displacement jump across the crack (or its derivative) and the applied loading to the crack faces (or its derivative). The choice of which equation is most suitable for computing the displacement jump in any particular configuration depends upon the given loading ($\jump{p}$ and $\av{p}$), the extent of interface imperfection as described by the parameter $\kappa$, and which quantities are sought.

It is possible to show (see Appendix \ref{appendix:limitcase}) that in the limit $\kappa \to 0$, the singular integral identities (\ref{res1})--(\ref{res2}) reduce to the known perfect interface case, see \cite{Piccolroaz2013}; this formulation is thus suitable for small $\kappa$. Formulations (\ref{res1b}) and (\ref{res3b}) are however less desirable for computations when $\kappa$ is small. It can be shown that in the limit $\kappa\to0$, the operator $-\frac{\xi_0}{\pi}\mT_{\xi_0}^{(s)}$ tends to the identity operator (see Theorem \ref{teo2} in Appendix \ref{appendix:limitcase}). Rewriting (\ref{res3b}) as
\begin{equation}
\label{res3bfactored}
-\frac{\xi_0}{\pi} \mT_{\xi_0}^{(s)} \av{p} + \frac{\mu_*\xi_0}{2\pi} \mT_{\xi_0}^{(s)} \jump{p}=\frac{1}{\kappa}\left(-\frac{\xi_0}{\pi} \mT_{\xi_0}^{(s)} \jump{u}^{(-)} -  \jump{u}^{(-)}\right),\quad x_1<0,
\end{equation}
and taking into account this limiting behaviour, the parentheses in the right hand side will be small for small $\kappa$, while the factor $1/\kappa$ will be large. Computations based on this formulation will therefore be sensitive to error for highly imperfect interfaces; in such a case the alternative formulations are more favourable.

\subsection{{Numerical examples}}

As an illustrative example, we present in Figure \ref{fig:jump} a plot of the normalised displacement jump $\jump{u_*}(x_1)$ for different values of the interface imperfection parameter $\kappa$. 
%
\begin{figure}[htcb!]
\centering
\includegraphics[width=\linewidth]{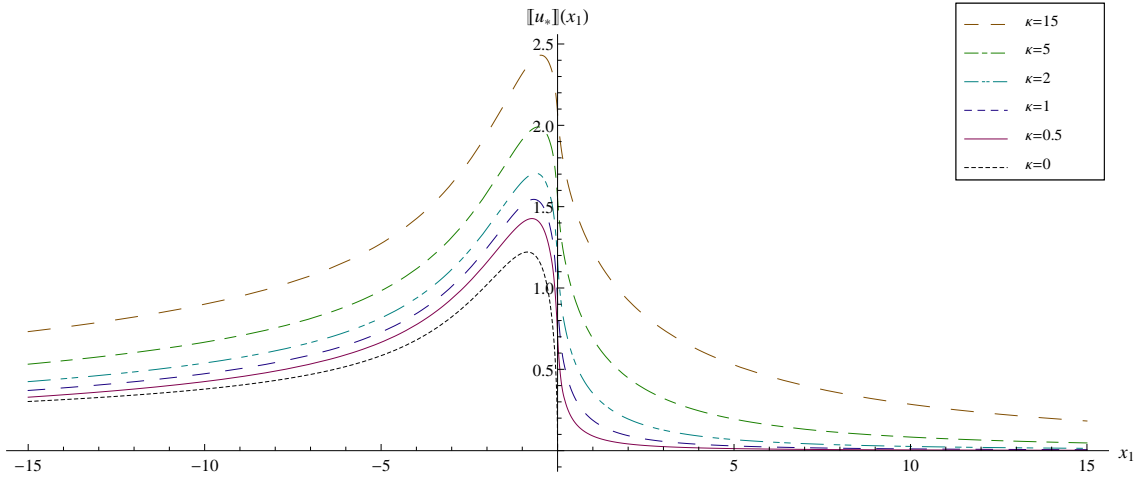}
\caption{\footnotesize Displacement jump across the crack and imperfect interface. For $x_1<0$, this has been computed using integral identity (\ref{res3b}), while for $x_1>0$, $\av{\sigma}^{(+)}$ has been computed and the displacement jump plotted via the relationship $\jump{u_*}^{(+)}=\kappa\av{\sigma_*}^{(+)}$. Also plotted is the displacement jump from the perfect interface case which corresponds to $\kappa=0$.}
\label{fig:jump}
\end{figure}
%
The crack faces have an applied smooth symmetric loading given by
\begin{equation}
\av{p}(x_1)=-\frac{T_0}{L}e^{x_1/L};\quad \jump{p}(x_1)=0,\quad x_1<0,
\end{equation}
and the materials above and below the crack share the same shear modulus (i.e. $\mu_1=\mu_2$); we emphasise however that asymmetric crack loadings and different materials can be used. The normalised displacement and traction variables are respectively defined as
\begin{equation}
u_*=\frac{\mu_1+\mu_2}{2T_0}u;\qquad \sigma_*=\frac{\sigma L}{T_0}. 
\end{equation}
The computations have been performed using a standard iterative procedure implemented in Mathematica to solve the integral equation (\ref{res3b}) in order to obtain $\jump{u_*}^{(-)}(x_1)$ (i.e. the displacement jump for $x_1<0$). The speed of convergence towards the solution is slower for smaller $\kappa$, as expected following our comments in the previous section regarding (\ref{res3b}) not being the most ideal formulation as $\kappa\to0$.
Figure \ref{fig:jump} also includes the displacement jump across the crack in the {\em perfect} interface case, as obtained in the paper of Piccolroaz et al. \cite{Piccolroaz2013}.
%
\begin{figure}[htcb!]
\begin{floatrow}
\capbtabbox{
\begin{tabular}{l||l|l}
$\kappa$ & $\jump{u_*}^{(-)}(0^-)$ & $\kappa\av{\sigma_*}^{(+)}(0^+)$ \\ \hline\hline
0.5      & 0.71625120          & 0.71625547                   \\
1        & 0.92723496          & 0.92723734                   \\
2        & 1.17477052          & 1.17477140                   \\
5        & 1.55561015          & 1.55561023                   \\
15       & 2.07698083          & 2.07698084                   \\ 
\end{tabular}\vspace{4em}
}
{\caption{\footnotesize Accuracy of agreement between $\jump{u_*}^{(-)}(0)$ and $\kappa\av{\sigma_*}^{(+)}(0)$ for different values of $\kappa$.}\label{table:values}}
\ffigbox{\includegraphics[width=1.\linewidth]{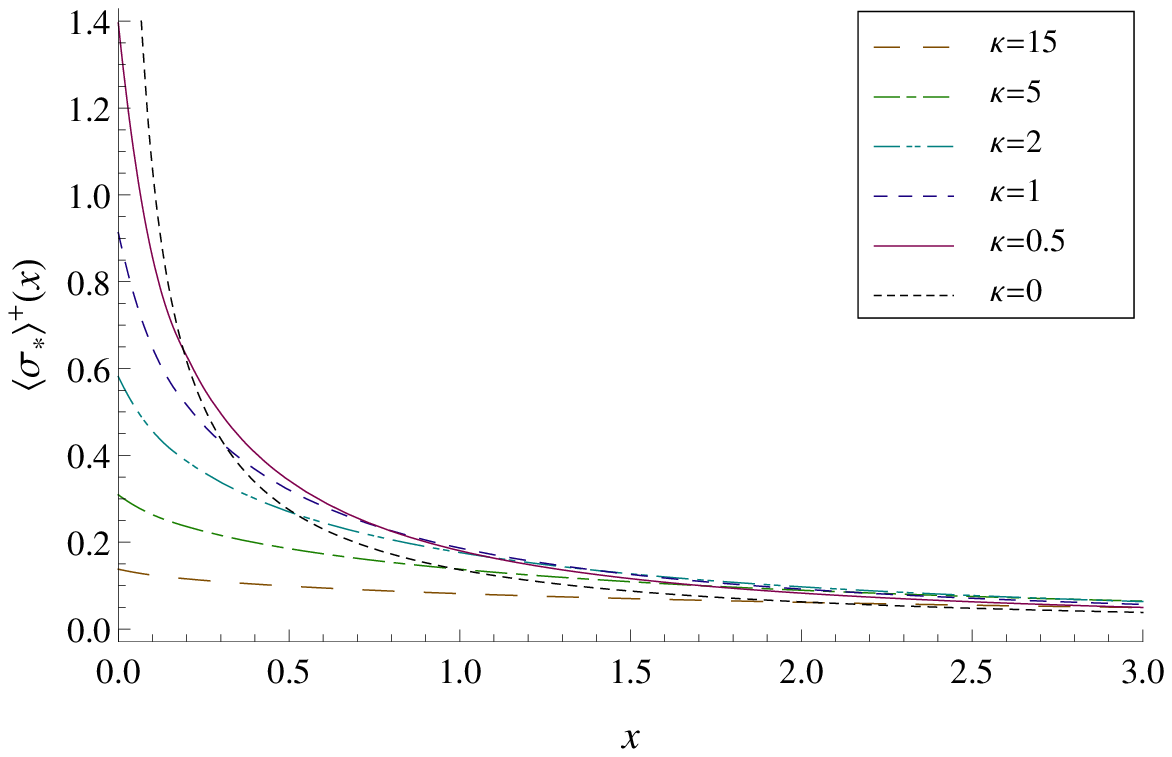}}{\caption{\footnotesize Distribution of tractions along the imperfect interface.}\label{fig:sigma}}
\end{floatrow}
\end{figure}
%

Having obtained $\jump{u_*}^{(-)}$, any of the four expressions for $\av{\sigma_*}^{(+)}$ can be used (since the loading is smooth) to obtain the traction along the imperfect interface ahead of the crack; the results of these computations are shown in Figure \ref{fig:sigma} and are also used to derive the displacement jump across the imperfect interface for $x_1>0$ in Figure \ref{fig:jump} via the relationship
\begin{equation}
\jump{u}^{(+)}(x_1)=\kappa\av{\sigma}^{(+)}(x_1).
\end{equation}
As expected, the tractions at the crack tip are higher for interfaces with a smaller extent of imperfection (i.e. for smaller $\kappa$). 
The tractions along a perfect interface (corresponding to the formulation with $\kappa=0$) is also shown in Figure \ref{fig:sigma}; this is computed via the integral relationships derived in Piccolroaz et al. \cite{Piccolroaz2013} and displays the usual square root singularity near the crack tip.

Near the crack tip, $\jump{u}(x_1)\sim\kappa\av{\sigma}(x_1)$, $x_1\to0^\pm$ \cite{Antipov2001,Mishuris2001}. This relationship provides a simple check of computational accuracy; values are presented in Table \ref{table:values} and display good accuracy. The crossing of the lines is explained by noting that in order to balance the prescribed exponential loading, the integral of $\av{\sigma_*}^{(+)}$ over the positive real axis must be equal to $T_0$. 

\begin{figure}[htcb!]
\centering
\includegraphics[width=\linewidth]{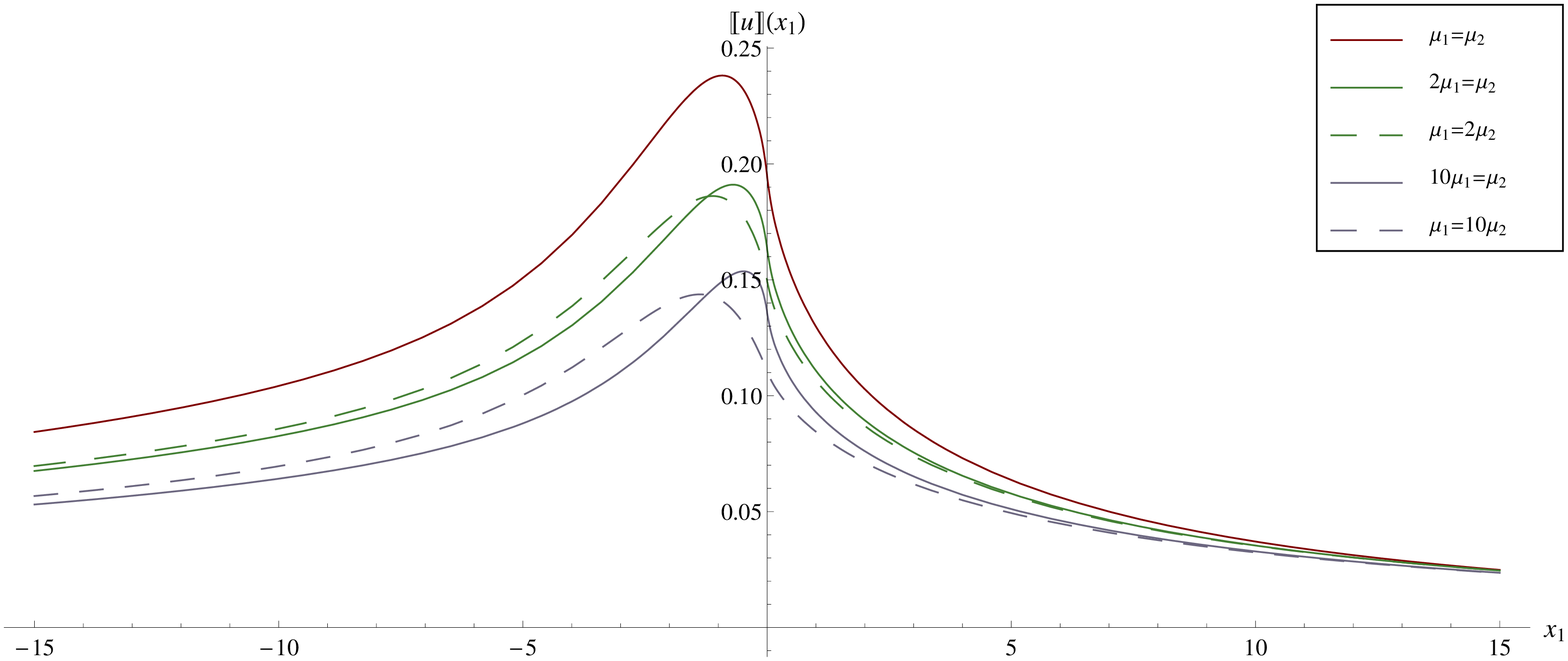}
\caption{\footnotesize Displacement jump across the crack and imperfect interface for an asymmetrically loaded interfacial crack, with varying contrast in stiffness of materials.}\label{fig:asymjump}
\end{figure}

\begin{figure}[htcb!]
\centering
\includegraphics[width=0.4\linewidth]{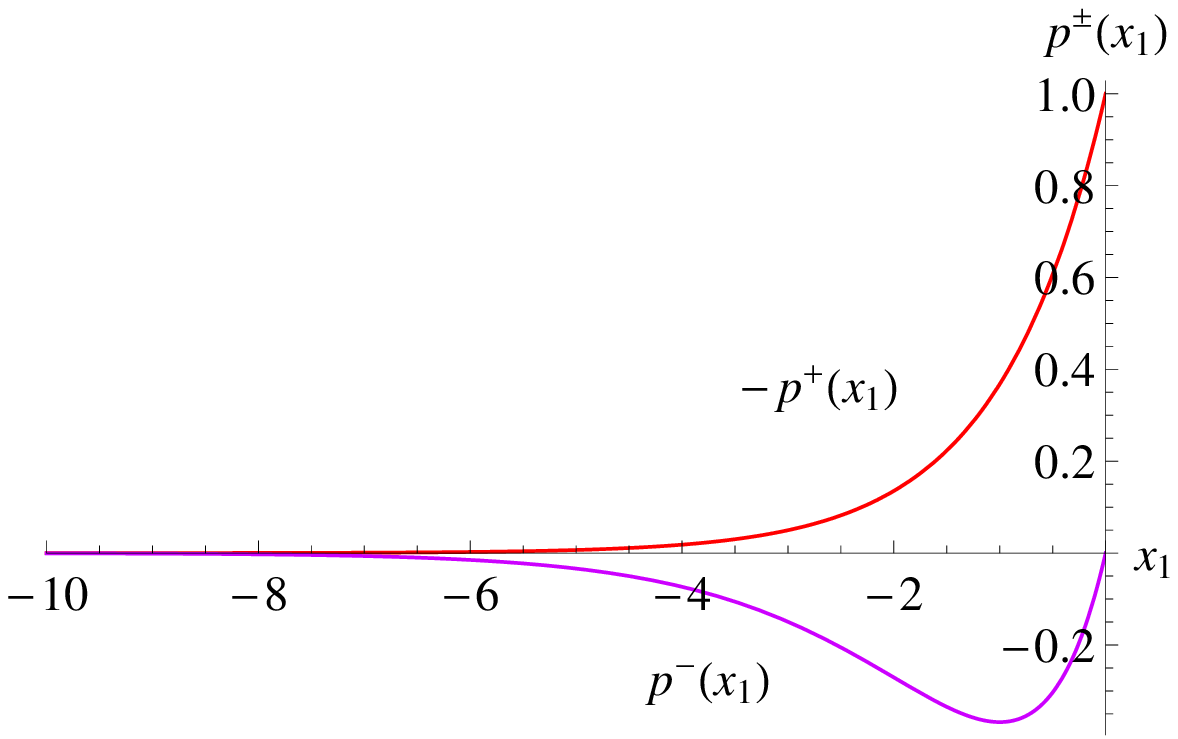}
\caption{\footnotesize Asymmetric applied loading of the form given in (\ref{asymloading}).}
\label{fig:asym_loading}
\end{figure}

In order to demonstrate the applicability of the identities to inhomogeneous bimaterials under asymmetric loads, Figure \ref{fig:asymjump} shows the displacement jump across the crack and the imperfect interface for a case with asymmetric self-balanced loadings applied to the crack faces of the form
\begin{equation}
\label{asymloading}
p_+(x_1)=-\frac{T_0}{L}e^{x_1/L};\qquad p_-(x_1)=\frac{T_0}{L^2}x_1e^{x_1/L};
\end{equation}
as shown in Figure \ref{fig:asym_loading}. The extent of imperfection of the interface used in the computations is $\kappa = 2$ in this example.

\section{Boundary integral equations for imperfect interface. Mode I and II}
\label{secmode12}

\subsection{Evaluation of the boundary integral equations}

In the case of plane strain deformation, the Betti identity (\ref{recid3}) relating the physical solution $\bu = (u_1,u_2)^T$, 
$\bsigma_2 = (\sigma_{21},\sigma_{22})^T$ with the weight function $\bU$, $\bSigma_2$ is given by
\begin{equation}
\label{m12id1}
\bR \jump{\bU} * \av{\bsigma_2}^{(+)} - \bR \av{\bSigma_2}^{(-)} * \jump{\bu} =
- \bR \jump{\bU} * \av{\bp} - \bR \av{\bU} * \jump{\bp},
\end{equation}
where $\av{\bp} = (\av{p_1}, \av{p_2})^T$, $\jump{\bp} = (\jump{p_1}, \jump{p_2})^T$ are the symmetric and skew-symmetric parts of the loading. Here and in the 
sequel of this section, we use the following matrices:
\begin{equation}
\label{Rmatrix}
\bR =
\begin{pmatrix}
-1 & 0 \\
0 & 1
\end{pmatrix}, \quad
\bI =
\begin{pmatrix}
1 & 0 \\
0 & 1
\end{pmatrix}, \quad
\bE =
\begin{pmatrix}
0 & 1 \\
-1 & 0
\end{pmatrix}.
\end{equation}
Note that the symmetric and skew-symmetric weight functions, $\jump{\bU}$ and $\av{\bU}$, and the corresponding traction $\av{\bSigma_2}$ are represented by 
2$\times$2 matrices. In fact, in the case of an elastic bimaterial plane, there are two linearly independent weight functions, 
$\bU^j = (U_1^j,U_2^j)^T$, $\bSigma_2^j = (\Sigma_{21}^j,\Sigma_{22}^j)^T$, $j = 1,2$, and it is possible to construct the weight function tensors by ordering 
the components of each weight function in columns of 2$\times$2 matrices \cite{Piccolroaz2009}.
\begin{equation}
\bU =
\begin{pmatrix}
U_1^1 & U_1^2 \\[1mm]
U_2^1 & U_2^2
\end{pmatrix}, \quad
\bSigma_2 =
\begin{pmatrix}
\Sigma_{21}^1 & \Sigma_{21}^2 \\[1mm]
\Sigma_{22}^1 & \Sigma_{22}^2
\end{pmatrix}.
\end{equation}
Applying the Fourier transform to the equation \eqref{m12id1}, we obtain
\begin{equation}
\label{m12id2}
(\jump{\tilde{\bU}}(\xi))^T \bR \av{\tilde{\bsigma}_2}^+(\xi) - (\av{\tilde{\bSigma}_2}^-(\xi))^T \bR \jump{\tilde{\bu}}(\xi) = 
-(\jump{\tilde{\bU}}(\xi))^T \bR \av{\tilde{\bp}}(\xi) - (\av{\tilde{\bU}}(\xi))^T \bR \jump{\tilde{\bp}}(\xi).
\end{equation}
We now split $\jump{\tilde{\bU}}$ into the sum of $\jump{\tilde{\bU}}^\pm$ and similarly split $\jump{\tilde{\bu}}$ into the sum of $\jump{\tilde{\bu}}^\pm$
\begin{multline}
\label{m12id3}
(\jump{\tilde{\bU}}^+(\xi))^T \bR \av{\tilde{\bsigma}_2}^+(\xi) + (\jump{\tilde{\bU}}^-(\xi))^T \bR \av{\tilde{\bsigma}_2}^+(\xi) - 
(\av{\tilde{\bSigma}_2}^-(\xi))^T \bR \jump{\tilde{\bu}}^+(\xi) - \\ 
(\av{\tilde{\bSigma}_2}^-(\xi))^T \bR \jump{\tilde{\bu}}^-(\xi) = 
-(\jump{\tilde{\bU}}(\xi))^T \bR \av{\tilde{\bp}}(\xi) - (\av{\tilde{\bU}}(\xi))^T \bR \jump{\tilde{\bp}}(\xi).
\end{multline}
We now make use of the transmission conditions which state that
\begin{equation}
\jump{\tilde{\bU}}^-(\xi) = \bK^*\ \av{\tilde{\bSigma}_2}^-(\xi), \quad
\jump{\tilde{\bu}}^+(\xi) = \bK \av{\tilde{\bsigma}_2}^+(\xi),
\end{equation}
where
\begin{equation}
\bK =
\begin{pmatrix}
K_{11} & K_{12} \\
K_{12} & K_{22}
\end{pmatrix}, \quad
\bK^* = \bR \bK^T \bR =
\begin{pmatrix}
K_{11} & -K_{12} \\
-K_{12} & K_{22}
\end{pmatrix}.
\end{equation}
This causes the second and third terms in the left hand side of \eqref{m12id3} to cancel, leaving
\begin{equation}
\label{m12id4}
(\jump{\tilde{\bU}}^+(\xi))^T \bR \av{\tilde{\bsigma}_2}^+(\xi) - (\av{\tilde{\bSigma}_2}^-(\xi))^T \bR \jump{\tilde{\bu}}^-(\xi) = 
-(\jump{\tilde{\bU}}(\xi))^T \bR \av{\tilde{\bp}}(\xi) - (\av{\tilde{\bU}}(\xi))^T \bR \jump{\tilde{\bp}}(\xi).
\end{equation}
Multiplying both sides by $\bR^{-1} (\jump{\tilde{\bU}}^+(\xi))^{-T}$ we get
\begin{equation}
\label{m12id5}
\av{\tilde{\bsigma}_2}^+(\xi) - \bB(\xi) \frac{\xi}{i} \jump{\tilde{\bu}}^-(\xi) =
- \bC(\xi) \av{\tilde{\bp}}(\xi) - \bA(\xi) \jump{\tilde{\bp}}(\xi),
\end{equation}
where $\bA(\xi)$, $\bB(\xi)$ and $\bC(\xi)$ are the following matrices
\begin{equation}
\bA = \bR^{-1} (\jump{\tilde{\bU}}^+)^{-T} \av{\tilde{\bU}}^T \bR, \quad
\bB = \frac{i}{\xi} \bR^{-1} (\jump{\tilde{\bU}}^+)^{-T} (\av{\tilde{\bSigma}_2}^-)^T \bR, \quad
\bC = \bR^{-1} (\jump{\tilde{\bU}}^+)^{-T} \jump{\tilde{\bU}}^T \bR,
\end{equation}
which can be computed using results for the symmetric and skew-symmetric weight functions obtained by Antipov \cite{Antipov1999} and Piccolroaz et al. \cite{Piccolroaz2009}. Namely
\begin{equation}
\jump{\tilde{\bU}}(\xi) = -\frac{1}{|\xi|} \big[b \bI - id\sign(\xi) \bE\big] \av{\tilde{\bSigma}_2}^-(\xi) =
-\frac{1}{|\xi|} \bG(\xi) \av{\tilde{\bSigma}_2}^-(\xi),
\end{equation}
\begin{equation}
\av{\tilde{\bU}}(\xi) = - \frac{b}{2|\xi|} \big[\alpha \bI - i\gamma\sign(\xi) \bE\big] \av{\tilde{\bSigma}_2}^-(\xi) =
-\frac{1}{2|\xi|} \bF(\xi) \av{\tilde{\bSigma}_2}^-(\xi),
\end{equation}
where $b$, $d$, $\alpha$ and $\gamma$ are the following bimaterial constants:
\begin{equation}
b = \frac{1 - \nu_1}{\mu_1} + \frac{1 - \nu_2}{\mu_2}, \quad d = \frac{1 - 2\nu_1}{2\mu_1} - \frac{1 - 2\nu_2}{2\mu_2},
\end{equation}
\begin{equation}
\alpha = \frac{\mu_2(1 - \nu_1) - \mu_1(1 - \nu_2)}{\mu_2(1 - \nu_1) + \mu_1(1 - \nu_2)}, \quad
\gamma = \frac{\mu_2(1 - 2\nu_1) + \mu_1(1 - 2\nu_2)}{2\mu_2(1 - \nu_1) + 2\mu_1(1 - \nu_2)}.
\end{equation}
As a result, we obtain
\begin{equation}
\bA(\xi) =
\frac{1}{2}
\bR^{-1} \Big[ |\xi|\bK^* + b \bI - id\sign(\xi) \bE \Big]^{-T} \Big[ b\alpha \bI - ib\gamma\sign(\xi) \bE \Big]^T \bR,
\end{equation}
\begin{equation}
\bB(\xi) =
-i\bR^{-1} \Big[ \xi\bK^* + b\sign(\xi) \bI - id \bE \Big]^{-T} \bR,
\end{equation}
\begin{equation}
\bC(\xi) =
\bR^{-1} \Big[ |\xi|\bK^* + b \bI - id\sign(\xi) \bE \Big]^{-T} \Big[ b \bI - id\sign(\xi) \bE \Big]^T \bR.
\end{equation}
Matrices $\bA(\xi)$, $\bB(\xi)$ and $\bC(\xi)$ are reported in Appendix \ref{matrices}.

Inverting the Fourier transform in \eqref{m12id5} for the two cases $x_1 < 0$ and $x_1 > 0$, we get
\begin{equation}
\label{m12id5a1}
\mF^{-1}_{x_1 < 0}\Big[\bC(\xi) \av{\tilde{\bp}}(\xi)\Big] + \mF^{-1}_{x_1 < 0}\Big[\bA(\xi) \jump{\tilde{\bp}}(\xi)\Big] =
\mF^{-1}_{x_1 < 0}\Big[\bB(\xi) \frac{\xi}{i}\jump{\tilde{\bu}}^-(\xi)\Big],
\end{equation}
\begin{equation}
\label{m12id5a2}
\av{\bsigma_2}^{(+)}(x_1) =
\mF^{-1}_{x_1 > 0}\Big[\bB(\xi) \frac{\xi}{i}\jump{\tilde{\bu}}^-(\xi)\Big] - \mF^{-1}_{x_1 > 0}\Big[\bC(\xi) \av{\tilde{\bp}}(\xi)\Big]
- \mF^{-1}_{x_1 > 0}\Big[\bA(\xi) \jump{\tilde{\bp}}(\xi)\Big].
\end{equation}
Similarly to the previous section, the term $\av{\tilde{\bsigma}_2}^+$ in \eqref{m12id5} cancels from (\ref{m12id5a1}) because it is a ``$+$'' function.
To proceed further, we need to perform the Fourier inversion of the matrices $\bA(\xi)$, $\bB(\xi)$ and $\bC(\xi)$. This is done in Appendix \ref{inversion}.

Finally, the integral equations for plane strain deformation in the imperfect interface case become
\begin{multline}
\label{resPS1}
\bmC^{(s)}\av{\bp}(x_1) + \bmA^{(s)}\jump{\bp}(x_1) = \bmB^{(s)}\frac{\partial \jump{\bu}^{(-)}}{\partial x_1} 
+ \frac{1}{\pi d_2(\xi_2 - \xi_1)} \sum_{j = 1}^{2} \bB_R^{(j)} T_{\xi_j}(x_1) \jump{\bu}^{(-)}(0^-) \\ 
+ \frac{1}{\pi d_2(\xi_2 - \xi_1)} \sum_{j = 1}^{2} \bB_I^{(j)} S_{\xi_j}(x_1) \jump{\bu}^{(-)}(0^-), \quad x_1 < 0,
\end{multline}
\begin{multline}
\label{resPS2}
\av{\bsigma_2}^{(+)}(x_1) = \bmB^{(c)}\frac{\partial \jump{\bu}^{(-)}}{\partial x_1} - \bmC^{(c)}\av{\bp}(x_1) - \bmA^{(c)}\jump{\bp}(x_1) 
+ \frac{1}{\pi d_2(\xi_2 - \xi_1)} \sum_{j = 1}^{2} \bB_R^{(j)} T_{\xi_j}(x_1) \jump{\bu}^{(-)}(0^-) \\ 
+ \frac{1}{\pi d_2(\xi_2 - \xi_1)} \sum_{j = 1}^{2} \bB_I^{(j)} S_{\xi_j}(x_1) \jump{\bu}^{(-)}(0^-), \quad x_1 > 0,
\end{multline}
where $\bmA^{(s)}$, $\bmB^{(s)}$, $\bmC^{(s)}$ are singular matrix operators, whereas $\bmA^{(c)}$, $\bmB^{(c)}$, $\bmC^{(c)}$ are {combinations of singular, weakly singular and fixed point singular} matrix operators, 
defined as
\begin{equation}
\bmA^{(s,c)} = -\frac{b}{2\pi d_2(\xi_2 - \xi_1)} \left\{ \sum_{j = 1}^{2} \bA_R^{(j)} \mT^{(s,c)}_{\xi_j} + \sum_{j = 1}^{2} \bA_I^{(j)} \mS^{(s,c)}_{\xi_j} \right\},
\end{equation}
\begin{equation}
\bmB^{(s,c)} = -\frac{1}{\pi d_2(\xi_2 - \xi_1)} \left\{ \sum_{j = 1}^{2} \bB_R^{(j)} \mT^{(s,c)}_{\xi_j} + \sum_{j = 1}^{2} \bB_I^{(j)} \mS^{(s,c)}_{\xi_j} \right\},
\end{equation}
\begin{equation}
\bmC^{(s,c)} = -\frac{1}{\pi d_2(\xi_2 - \xi_1)} \left\{ \sum_{j = 1}^{2} \bC_R^{(j)} \mT^{(s,c)}_{\xi_j} + \sum_{j = 1}^{2} \bC_I^{(j)} \mS^{(s,c)}_{\xi_j} \right\}.
\end{equation}
Here, $\mT^{(s,c)}_{\xi_j}$ and $\mS^{(s,c)}_{\xi_j}$ are the operators defined in eq. (\ref{ops}) and (\ref{opc}); the constants $\xi_j$ are given in Appendix  \ref{inversion}. The constant matrices $A_{R,I}^{(j)}$, 
$B_{R,I}^{(j)}$ and $C_{R,I}^{(j)}$ are also given in Appendix \ref{inversion}, see eqs. (\ref{Ai})--(\ref{Af}), eqs. 
(\ref{Bi})--(\ref{Bf}) and eqs. (\ref{Ci})--(\ref{Cf}), respectively.

Note that in (\ref{resPS1}) and (\ref{resPS2}) the function $T_{\xi_j}(x_1)$ is singular at $x_1 = 0$, but
\begin{equation}
\sum_{j = 1}^{2} \bB_R^{(j)} T_{\xi_j}(x_1) = d\bE (T_{\xi_1}(x_1) - T_{\xi_2}(x_2)),
\end{equation}
and $T_{\xi_1}(x_1) - T_{\xi_2}(x_2)$ is not singular at $x_1 = 0$ since
\begin{equation}
T_{\xi_1}(x_1) - T_{\xi_2}(x_2) = \log\left(\frac{\xi_1}{\xi_2}\right) + O(|x_1|), \quad x_1 \to 0.
\end{equation}

\subsection{Alternative integral formulae}

In a similar spirit to the alternative formulae derived for Mode III, alternative integral formulae for the plane strain integral equations (\ref{resPS1}) and (\ref{resPS2}) can be derived  by using the relationship between $\mS_{\xi_j}$ and $\mT_{\xi_j}$ given in (\ref{identi}), along with a further relationship
\begin{equation}
\mT_{\xi_j}\varphi'=\int\limits_{-\infty}^\infty\frac{\varphi(t)}{x-t}\mathrm{d}t + \xi_j \mS_{\xi_j}\varphi
\end{equation}
which results from integrating $\mT_{\xi_j}$ by parts. These relationships allow $\mT_{\xi_j}\varphi'$ to be expressed in terms of $\mS_{\xi_j}\varphi$, and $\mS_{\xi_j}\varphi'$ to be expressed in terms of $\mT_{\xi_j}\varphi$.

\section{Conclusions}

Boundary integral formulations relating the applied loading and the resulting crack opening have been derived for a semi-infinite crack sitting along a soft imperfect interface. In contrast to the Cauchy-type kernel from the perfect interface case, the presence of the imperfect interface introduces a weak logarithmic singularity to the integral operator. Moreover, alternative integral formulae have been derived. A choice of which formulation to employ, based upon the loading and problem parameters of the specific configuration, can lead to improved computational ease and efficiency. The identities derived here could be used in the modelling of hydraulic fracture with existing faults (interfaces) in the rock \cite{Adachi2007,Spence1985}.

\vspace{10mm}
{\bf Acknowledgement.} G.M. and A.V. gratefully acknowledge support from the European Union Seventh Framework Programme under respective contract numbers PIAP-GA-2009-251475 and PIAP-GA-2011-286110. 
AP would like to acknowledge the Italian Ministry of Education, University and Research (MIUR) for the grant FIRB 2010 Future in Research ``Structural mechanics 
models for renewable energy applications'' (RBFR107AKG).

\bibliographystyle{jabbrv_unsrt}
\bibliography
{%
../../BIBTEX/roaz}

\clearpage
\appendix
\renewcommand{\theequation}{\thesection.\arabic{equation}}

\section{APPENDIX}
\setcounter{equation}{0}

\subsection{Matrices $\bA(\xi)$, $\bB(\xi)$ and $\bC(\xi)$}
\label{matrices}

Matrices $\bA(\xi)$, $\bB(\xi)$ and $\bC(\xi)$ admit the following representation:
\begin{equation}
\bA(\xi) = \frac{b}{2D}
\begin{pmatrix}
A_{11} & A_{12} \\
A_{21} & A_{22}
\end{pmatrix}, \quad
\bB(\xi) = \frac{1}{D}
\begin{pmatrix}
B_{11} & B_{12} \\
B_{21} & B_{22}
\end{pmatrix}, \quad
\bC(\xi) = \frac{1}{D}
\begin{pmatrix}
C_{11} & C_{12} \\
C_{21} & C_{22}
\end{pmatrix}
\end{equation}
where the denominator $D$ is defined as
\begin{equation}
\label{denominator}
D = d_0 + d_1 |\xi| + d_2 |\xi|^2,
\end{equation}
\begin{equation}
d_0 = b^2 - d^2, \quad
d_1 = b(K_{11} + K_{22}) = b \tr \bK, \quad
d_2 = K_{11}K_{22} - (K_{12})^2 = \det \bK,
\end{equation}
and the elements $A_{ij}$, $B_{ij}$, $C_{ij}$ are given by
\begin{align}
A_{11} &= b\alpha - d\gamma + (\alpha K_{22} - i\gamma K_{12} \sign(\xi))|\xi|, \\
A_{12} &= i(d\alpha - b\gamma)\sign(\xi) - (\alpha K_{12} + i\gamma K_{22} \sign(\xi))|\xi|, \\
A_{21} &= -i(d\alpha - b\gamma)\sign(\xi) - (\alpha K_{12} - i\gamma K_{11} \sign(\xi))|\xi|, \\
A_{22} &= b\alpha - d\gamma + (\alpha K_{11} + i\gamma K_{12} \sign(\xi))|\xi|,
\end{align}
\begin{align}
B_{11} &= -ib\sign(\xi) - i K_{22} \xi, \\
B_{12} &= d + i K_{12} \xi, \\
B_{21} &= -d + i K_{12} \xi, \\
B_{22} &= -ib\sign(\xi) - i K_{11} \xi,
\end{align}
\begin{align}
C_{11} &= b^2 - d^2 + (b K_{22} - id K_{12} \sign(\xi))|\xi|, \\
C_{12} &= -(b K_{12} + id K_{22} \sign(\xi))|\xi|, \\
C_{21} &= -(b K_{12} - id K_{11} \sign(\xi))|\xi|, \\
C_{22} &= b^2 - d^2 + (b K_{11} + id K_{12} \sign(\xi))|\xi|.
\end{align}

\subsection{Fourier inversion of matrices $\bA(\xi)$, $\bB(\xi)$ and $\bC(\xi)$}
\label{inversion}

\paragraph{General procedure for the Fourier inversion.}

In order to perform the Fourier inversion of the matrices $\bA(\xi)$, $\bB(\xi)$ and $\bC(\xi)$, we first factorize the denominator $D$ defined in 
(\ref{denominator}) as follows
\begin{equation}
D = d_2 (|\xi| + \xi_1) (|\xi| + \xi_2),
\end{equation}
where
\begin{equation}
\label{xi12}
\xi_{1,2} = \frac{d_1 \mp \sqrt{d_1^2 - 4d_2d_0}}{2d_2} > 0,
\end{equation}
The typical term to invert is of the form
\begin{equation}
F(\xi) = \frac{F_R + F_R^\dag|\xi|}{D} + i \frac{F_I \sign(\xi) + F_I^\dag\xi}{D},
\end{equation}
Note that the function $F$ has the following property
\begin{equation}
F(-\xi) = \overline{F(\xi)},
\end{equation}
so that the Fourier inversion can be obtained as
\begin{equation}
\mF^{-1}[F(\xi)] = \frac{1}{\pi} \Re \int_0^\infty F(\xi) e^{-ix\xi} d\xi, 
= \frac{1}{\pi} \int_0^\infty \Re[F(\xi)] \cos(x\xi) d\xi  
+ \frac{1}{\pi} \int_0^\infty \Im[F(\xi)] \sin(x\xi) d\xi,
\end{equation}
where for $\xi > 0$
\begin{equation}
\Re[F(\xi)] = \frac{F_{R} + F_{R}^\dag\xi}{D} = \sum_{j = 1}^2 \frac{F_{R}^{(j)}}{d_2(\xi_2 - \xi_1)(\xi + \xi_j)},
\end{equation}
\begin{equation}
\Im[F(\xi)] = \frac{F_{I} + F_{I}^\dag\xi}{D} = \sum_{j = 1}^2 \frac{F_{I}^{(j)}}{d_2(\xi_2 - \xi_1)(\xi + \xi_j)},
\end{equation}
and
\begin{equation}
F_{R,I}^{(1)} = F_{R,I} - F_{R,I}^\dag\xi_1, \quad F_{R,I}^{(2)} = -F_{R,I} + F_{R,I}^\dag\xi_2.
\end{equation}
Now we can make use of the following formulae
\begin{equation}
\int_0^\infty \Re[F(\xi)] \cos(x\xi) d\xi = \sum_{j = 1}^{2} \frac{F_R^{(j)}}{d_2(\xi_2 - \xi_1)} \int_0^\infty \frac{\cos(x\xi)}{\xi + \xi_j} d\xi = 
-\frac{1}{d_2(\xi_2 - \xi_1)} \sum_{j = 1}^{2} F_R^{(j)} T_{\xi_j}(x),
\end{equation}
\begin{equation}
\int_0^\infty \Im[F(\xi)] \sin(x\xi) d\xi = \sum_{j = 1}^{2} \frac{F_I^{(j)}}{d_2(\xi_2 - \xi_1)} \int_0^\infty \frac{\sin(x\xi)}{\xi + \xi_j} d\xi = 
-\frac{1}{d_2(\xi_2 - \xi_1)} \sum_{j = 1}^{2} F_I^{(j)} S_{\xi_j}(x),
\end{equation}
where functions $S_{\xi_j}(x)$ and $T_{\xi_j}(x)$ are defined as in (\ref{S}) and (\ref{T}), respectively.

Finally we obtain the Fourier inversion of the general term $F(\xi)$ as follows
\begin{equation}
\mF^{-1}[F(\xi)] = -\frac{1}{\pi d_2(\xi_2 - \xi_1)} \left\{ \sum_{j = 1}^{2} F_R^{(j)} T_{\xi_j}(x) 
+ \sum_{j = 1}^{2} F_I^{(j)} S_{\xi_j}(x) \right\}.
\end{equation}

\paragraph{Fourier inversion of the matrix $\bA(\xi)$.}

For $\xi > 0$ we can write 
\begin{equation}
\bA(\xi) = \frac{b}{2D}(\bA_R + \bA_R^\dag \xi) + \frac{ib}{2D} (\bA_I + \bA_I^\dag \xi), 
= \frac{b}{2d_2(\xi_2 - \xi_1)} \left\{ \sum_{j = 1}^{2} \frac{1}{\xi + \xi_j} \bA_R^{(j)} + 
i \sum_{j = 1}^{2} \frac{1}{\xi + \xi_j} \bA_I^{(j)} \right\},
\end{equation}
where
\begin{equation}
\bA_R = (b\alpha - d\gamma) \bI, \quad
\bA_R^\dag = \alpha
\begin{pmatrix}
K_{22} & -K_{12} \\
-K_{12} & K_{11}
\end{pmatrix},
\end{equation}
\begin{equation}
\bA_I = (d\alpha - b\gamma) \bE, \quad
\bA_I^\dag = \gamma
\begin{pmatrix}
-K_{12} & -K_{22} \\
K_{11} & K_{12}
\end{pmatrix},
\end{equation}
\begin{equation}
\label{Ai}
\bA_R^{(1)} = \bA_R - \bA_R^\dag \xi_1 = 
\begin{pmatrix}
b\alpha - d\gamma - \alpha \xi_1 K_{22} & \alpha \xi_1 K_{12} \\
\alpha \xi_1 K_{12} & b\alpha - d\gamma - \alpha \xi_1 K_{11}
\end{pmatrix},
\end{equation}
\begin{equation}
\bA_R^{(2)} = -\bA_R + \bA_R^\dag \xi_2 =
\begin{pmatrix}
-b\alpha + d\gamma + \alpha \xi_2 K_{22} & -\alpha \xi_2 K_{12} \\
-\alpha \xi_2 K_{12} & -b\alpha + d\gamma + \alpha \xi_2 K_{11}
\end{pmatrix},
\end{equation}
\begin{equation}
\bA_I^{(1)} = \bA_I - \bA_I^\dag \xi_1 = 
\begin{pmatrix}
\gamma \xi_1 K_{12} & d\alpha - b\gamma + \gamma \xi_1 K_{22} \\
-d\alpha + b\gamma - \gamma \xi_1 K_{11} & -\gamma \xi_1 K_{12}
\end{pmatrix},
\end{equation}
\begin{equation}
\label{Af}
\bA_I^{(2)} = -\bA_I + \bA_I^\dag \xi_2 =
\begin{pmatrix}
-\gamma \xi_2 K_{12} & -d\alpha + b\gamma - \gamma \xi_2 K_{22} \\
d\alpha - b\gamma + \gamma \xi_2 K_{11} & \gamma \xi_2 K_{12}
\end{pmatrix}.
\end{equation}
The Fourier inverse of the matrix $\bA(\xi)$ is then
\begin{equation}
\mF^{-1}[\bA(\xi)] = -\frac{b}{2\pi d_2(\xi_2 - \xi_1)} \left\{ \sum_{j = 1}^{2} \bA_R^{(j)} T_{\xi_j}(x) + \sum_{j = 1}^{2} \bA_I^{(j)} S_{\xi_j}(x) \right\}.
\end{equation}

\paragraph{Fourier inversion of the matrix $\bB(\xi)$.}

For $\xi > 0$ we can write 
\begin{equation}
\bB(\xi) = \frac{1}{D}(\bB_R + \bB_R^\dag \xi) + \frac{i}{D} (\bB_I + \bB_I^\dag \xi), 
= \frac{1}{d_2(\xi_2 - \xi_1)} \left\{ \sum_{j = 1}^{2} \frac{1}{\xi + \xi_j} \bB_R^{(j)} + 
i \sum_{j = 1}^{2} \frac{1}{\xi + \xi_j} \bB_I^{(j)} \right\},
\end{equation}
where
\begin{equation}
\bB_R = d \bE, \quad
\bB_R^\dag = \b0,
\end{equation}
\begin{equation}
\bB_I = -b \bI, \quad
\bB_I^\dag = 
\begin{pmatrix}
-K_{22} & K_{12} \\
K_{12} & -K_{11}
\end{pmatrix},
\end{equation}
\begin{equation}
\label{Bi}
\bB_R^{(1)} = \bB_R - \bB_R^\dag \xi_1 = 
d \bE,
\end{equation}
\begin{equation}
\bB_R^{(2)} = -\bB_R + \bB_R^\dag \xi_2 =
-d \bE,
\end{equation}
\begin{equation}
\bB_I^{(1)} = \bB_I - \bB_I^\dag \xi_1 = 
\begin{pmatrix}
-b + \xi_1 K_{22} & -\xi_1 K_{12} \\
-\xi_1 K_{12} & -b + \xi_1 K_{11}
\end{pmatrix},
\end{equation}
\begin{equation}
\label{Bf}
\bB_I^{(2)} = -\bB_I + \bB_I^\dag \xi_2 = 
\begin{pmatrix}
b - \xi_2 K_{22} & \xi_2 K_{12} \\
\xi_2 K_{12} & b - \xi_2 K_{11}
\end{pmatrix}.
\end{equation}
The Fourier inverse of the matrix $\bB(\xi)$ is then
\begin{equation}
\mF^{-1}[\bB(\xi)] = -\frac{1}{\pi d_2(\xi_2 - \xi_1)} \left\{ \sum_{j = 1}^{2} \bB_R^{(j)} T_{\xi_j}(x) + \sum_{j = 1}^{2} \bB_I^{(j)} S_{\xi_j}(x) \right\}.
\end{equation}

\paragraph{Fourier inversion of the matrix $\bC(\xi)$.}

For $\xi > 0$ we can write 
\begin{equation}
\bC(\xi) = \frac{1}{D}(\bC_R + \bC_R^\dag \xi) + \frac{i}{D} (\bC_I + \bC_I^\dag \xi), 
= \frac{1}{d_2(\xi_2 - \xi_1)} \left\{ \sum_{j = 1}^{2} \frac{1}{\xi + \xi_j} \bC_R^{(j)} + 
i \sum_{j = 1}^{2} \frac{1}{\xi + \xi_j} \bC_I^{(j)} \right\},
\end{equation}
where
\begin{equation}
\bC_R = (b^2 - d^2) \bI, \quad
\bC_R^\dag = b 
\begin{pmatrix}
K_{22} & -K_{12} \\
-K_{12} & K_{11}
\end{pmatrix},
\end{equation}
\begin{equation}
\bC_I = \b0, \quad
\bC_I^\dag = d 
\begin{pmatrix}
-K_{12} & -K_{22} \\
K_{11} & K_{12}
\end{pmatrix},
\end{equation}
\begin{equation}
\label{Ci}
\bC_R^{(1)} = \bC_R - \bC_R^\dag \xi_1 = 
\begin{pmatrix}
b^2 - d^2 - b\xi_1K_{22} & b\xi_1K_{12} \\
b\xi_1K_{12} & b^2 - d^2 - b\xi_1K_{11}
\end{pmatrix},
\end{equation}
\begin{equation}
\bC_R^{(2)} = -\bC_R + \bC_R^\dag \xi_2 =
\begin{pmatrix}
-b^2 + d^2 + b\xi_2K_{22} & -b\xi_2K_{12} \\
-b\xi_2K_{12} & -b^2 + d^2 + b\xi_2K_{11}
\end{pmatrix},
\end{equation}
\begin{equation}
\bC_I^{(1)} = \bC_I - \bC_I^\dag \xi_1 = -d\xi_1 
\begin{pmatrix}
-K_{12} & -K_{22} \\
K_{11} & K_{12}
\end{pmatrix},
\end{equation}
\begin{equation}
\label{Cf}
\bC_I^{(2)} = -\bC_I + \bC_I^\dag \xi_2 = d\xi_2
\begin{pmatrix}
-K_{12} & -K_{22} \\
K_{11} & K_{12}
\end{pmatrix}.
\end{equation}
The Fourier inverse of the matrix $\bC(\xi)$ is then
\begin{equation}
\mF^{-1}[\bC(\xi)] = -\frac{1}{\pi d_2(\xi_2 - \xi_1)} \left\{ \sum_{j = 1}^{2} \bC_R^{(j)} T_{\xi_j}(x) + \sum_{j = 1}^{2} \bC_I^{(j)} S_{\xi_j}(x) \right\}.
\end{equation}

\subsection{The limit case $\|\bK\| \to 0$. Perfect interface.}\label{appendix:limitcase}

It is possible to show that, in the limit $\|\bK\| \to 0$, the singular integral identities, (\ref{res1})--(\ref{res2}) for Mode III 
and (\ref{resPS1})--(\ref{resPS2}) for Mode I-II, reduce to the known case of perfect interface, see Piccolroaz and Mishuris \cite{Piccolroaz2013}. To this
purpose, we observe that the parameters $\xi_0$, defined in (\ref{xi0}) and $\xi_j$, $j = 1,2$, defined in (\ref{xi12}), tend to infinity as $\|\bK\| \to 0$, 
and thus we use the asymptotics of the sine and cosine integral functions for large argument $x \to +\infty$
\begin{equation}
\si(x) = -\frac{\cos x}{x} \Sigma_1(x) - \frac{\sin x}{x} \Sigma_2(x),
\end{equation}
\begin{equation}
\ci(x) = \frac{\sin x}{x} \Sigma_1(x) - \frac{\cos x}{x} \Sigma_2(x),
\end{equation}
where $\Sigma_1$ and $\Sigma_2$ are the following asymptotic series
\begin{equation}
\label{est1}
\Sigma_1(x) = 1 - \frac{2!}{x^2} + \cdots + \frac{2n!}{x^{2n}} + O\left(\frac{1}{x^{2n + 2}}\right), \quad x \to \infty
\end{equation}
\begin{equation}
\label{est2}
\Sigma_2(x) = \frac{1}{x} - \frac{3!}{x^3} + \cdots + \frac{(2n + 1)!}{x^{2n + 1}} + O\left(\frac{1}{x^{2n + 3}}\right), \quad x \to \infty.
\end{equation}
We also make use of the asymptotics of the sine and cosine integral functions for small argument $x \to 0$
\begin{equation}
\label{est3}
\si(x) = -\frac{\pi}{2} + x - \frac{x^3}{3 \cdot 3!} + \cdots + \frac{(-1)^n x^{2n + 1}}{(2n + 1)(2n + 1)!} + O(x^{2n + 3}), \quad x \to 0,
\end{equation}
\begin{equation}
\label{est4}
\ci(x) = \ln x + \gamma - \frac{x^2}{2 \cdot 2!} + \cdots + \frac{(-1)^n x^{2n}}{(2n)(2n)!} + O(x^{2n + 2}), \quad x \to 0.
\end{equation}
Note that the functions $S_{\xi}$ and $T_{\xi}$ defined in (\ref{S}) and (\ref{T}) can be written as
\begin{align}
S_{\xi}(x) &= \sign(x) \frac{\Sigma_1(\xi|x|)}{-\xi|x|}, \\[3mm]
T_{\xi}(x) &= \frac{\Sigma_2(\xi|x|)}{-\xi|x|}.
\end{align}

\begin{theorem}
\label{teo1}
In the limit case $\xi \to \infty$, the kernel of the operator $\frac{\xi}{\pi}\mS_{\xi}$ reduces to the kernel of the Cauchy type (in the class of functions 
satisfying the H\"older condition), so that
\begin{equation}
\lim_{\xi \to \infty} \frac{\xi}{\pi} \mS_{\xi} \varphi = -\frac{1}{\pi} \int_{-\infty}^{\infty} \frac{\varphi(t)}{x - t} dt, \quad \xi \to \infty.
\end{equation}
\end{theorem}

\begin{proof}
\begin{align}
\psi = \frac{\xi}{\pi}\mS_{\xi} \varphi
&= \frac{\xi}{\pi} \int_{-\infty}^{\infty} \sign(t) \frac{1}{-\xi|t|} \Sigma_1(\xi |t|) \varphi(x - t) dt, \\
&= -\frac{1}{\pi} \int_{-\infty}^{\infty} \frac{1}{t} \Sigma_1(\xi t) \varphi(x - t) dt, \\
&= -\frac{1}{\pi} \int_{-\infty}^{\infty} \frac{1}{\tau} \Sigma_1(\tau) \varphi\left(x - \frac{\tau}{\xi}\right) d\tau, \\
&= -\frac{1}{\pi} \int_{-\infty}^{\infty} \frac{1}{\tau} \varphi\left(x - \frac{\tau}{\xi}\right) d\tau
   -\frac{1}{\pi} \int_{-\infty}^{\infty} \frac{\Sigma_1(\tau) - 1}{\tau} \varphi\left(x - \frac{\tau}{\xi}\right) d\tau, 
\label{int12}
\end{align}
The first integral in \eqref{int12}, after substituting back $\tau = \xi t$, gives
\begin{equation}
I_1 = -\frac{1}{\pi}\int_{-\infty}^{\infty} \frac{1}{t} \varphi(x - t) dt = -\frac{1}{\pi} \int_{-\infty}^{\infty} \frac{\varphi(t)}{x - t} dt,
\end{equation}
whereas the second integral can be shown to vanish as $\xi \to \infty$. In fact, since $\Sigma_1$ is even
\begin{align}
\label{sec}
I_2 = \int_{-\infty}^{\infty} \frac{\Sigma_1(\tau) - 1}{\tau} \varphi\left(x - \frac{\tau}{\xi}\right) d\tau
&= \int_{0}^{\infty} \frac{\Sigma_1(\tau) - 1}{\tau} \big[\varphi\left(x - \frac{\tau}{\xi}\right) - \varphi\left(x + \frac{\tau}{\xi}\right)\big] d\tau, \\
&= \frac{1}{\xi^\lambda} \int_{0}^{\infty} \frac{\Sigma_1(\tau) - 1}{\tau^{1 - \lambda}}
\frac{\left[\varphi\left(x - \frac{\tau}{\xi}\right) - \varphi\left(x + \frac{\tau}{\xi}\right)\right]}{\left(\frac{\tau}{\xi}\right)^\lambda} d\tau,
\end{align}
where $0 < \lambda < 1$. Taking the absolute value in \eqref{sec}, we obtain the following bound
\begin{align}
|I_2|
&\leq \frac{1}{\xi^\lambda} \int_{0}^{\infty} \left|\frac{\Sigma_1(\tau) - 1}{\tau^{1 - \lambda}}\right|
      \left|\frac{\left[\varphi\left(x - \frac{\tau}{\xi}\right) - \varphi\left(x + \frac{\tau}{\xi}\right)\right]}{\left(\frac{\tau}{\xi}\right)^\lambda}\right| d\tau, \\
&\leq \frac{M}{\xi^\lambda} \int_{0}^{\infty} \left|\frac{\Sigma_1(\tau) - 1}{\tau^{1 - \lambda}}\right| d\tau, 
\label{hol}
\end{align}
where we used the H\"older condition for the function $\varphi$
\begin{equation}
\left|\frac{\left[\varphi\left(x - \frac{\tau}{\xi}\right) - \varphi\left(x + \frac{\tau}{\xi}\right)\right]}{\left(\frac{\tau}{\xi}\right)^\lambda}\right| \leq M.
\end{equation}
Note that the integral in the RHS of \eqref{hol} is converging since, from the estimates (\ref{est1}),(\ref{est3}) and (\ref{est4}), we have
\begin{equation}
\frac{\Sigma_1(\tau) - 1}{\tau^{1 - \lambda}} \sim -\frac{2}{\tau^{3 - \lambda}}, \quad \tau \to \infty,
\end{equation}
\begin{equation}
\frac{\Sigma_1(\tau) - 1}{\tau^{1 - \lambda}} \sim -\frac{1}{\tau^{1 - \lambda}}, \quad \tau \to 0.
\end{equation}
Taking the limit $\xi \to \infty$, we obtain $I_2 \to 0$, which concludes the proof.
\end{proof}

\begin{theorem}
\label{teo2}
In the limit case $\xi \to \infty$, the kernel of the operator $-\frac{\xi}{\pi}\mT_{\xi}$ reduces to the Dirac delta function (in the class of functions satisfying the H\"older condition),
so that $-\frac{\xi}{\pi}\mT_{\xi}^{(s)}$ is reduced to the identity operator, whereas $-\frac{\xi}{\pi}\mT_{\xi}^{(c)}$ is reduced to the null operator.
\end{theorem}

\begin{proof}
\begin{align}
\psi = -\frac{\xi}{\pi}\mT_{\xi} \varphi
&= -\frac{\xi}{\pi} \int_{-\infty}^{\infty} \frac{1}{-\xi|t|} \Sigma_2(\xi |t|) \varphi(x - t) dt, \\
&= \frac{1}{\pi} \int_{-\infty}^{\infty} \frac{1}{|t|} \Sigma_2(\xi |t|) \varphi(x - t) dt, \\
&= \frac{1}{\pi}\int_{-\infty}^{\infty} \frac{1}{|\tau|} \Sigma_{2}(|\tau|) \varphi\left(x - \frac{\tau}{\xi}\right) d\tau, \\
&= \frac{\varphi(x)}{\pi}\int_{-\infty}^{\infty} \frac{1}{|\tau|} \Sigma_2(|\tau|) d\tau
   + \frac{1}{\pi} \int_{-\infty}^{\infty} \frac{1}{|\tau|}\Sigma_2(|\tau|) \left[\varphi\left(x - \frac{\tau}{\xi}\right) - \varphi(x)\right] d\tau. \label{int12b}
\end{align}
The first integral in \eqref{int12b} gives
\begin{equation}
I_1 = \frac{\varphi(x)}{\pi}\int_{-\infty}^{\infty} \frac{1}{|\tau|} \Sigma_2(|\tau|) d\tau = \frac{\varphi(x)}{\pi} \pi = \varphi(x),
\end{equation}
whereas the second integral can be shown to vanish as $\xi \to \infty$. In fact, taking the absolute value of the integral we have the following bound
\begin{align}
|I_2|
&\leq \frac{1}{\pi} \int_{-\infty}^{\infty} \frac{|\Sigma_2(|\tau|)|}{|\tau|} \left|\varphi\left(x - \frac{\tau}{\xi}\right) - \varphi(x)\right| d\tau \\
&\leq \frac{1}{\pi \xi^\lambda} \int_{-\infty}^{\infty} \frac{|\Sigma_2(|\tau|)|}{|\tau|^{1 - \lambda}}
\frac{\left|\varphi\left(x - \frac{\tau}{\xi}\right) - \varphi(x)\right|}{\left(\frac{|\tau|}{\xi}\right)^\lambda} d\tau \\
&\leq \frac{M}{\pi \xi^\lambda} \int_{-\infty}^{\infty} \frac{|\Sigma_2(|\tau|)|}{|\tau|^{1 - \lambda}} d\tau, \label{hol2}
\end{align}
where again the H\"older condition for the function $\varphi$ has been used.

Note that the integral in the RHS of \eqref{hol2} is converging since, from the estimates (\ref{est2}),(\ref{est3}) and (\ref{est4}), we have
\begin{equation}
\frac{|\Sigma_2(|\tau|)|}{|\tau|^{1 - \lambda}} \sim \frac{1}{|\tau|^{2 - \lambda}}, \quad \tau \to \infty,
\end{equation}
\begin{equation}
\frac{|\Sigma_2(|\tau|)|}{|\tau|^{1 - \lambda}} \sim |\tau|^\lambda \ln|\tau|, \quad \tau \to 0.
\end{equation}
\end{proof}

\paragraph{Mode III.} Using the results of Theorems \ref{teo1} and \ref{teo2} it is easy to show that in case of perfect interface, $\|\bK\| \to 0$, the integral identities for antiplane deformation reduce to 
\begin{equation}
\langle p \rangle(x_1) - \frac{\mu_*}{2} \jump{p}(x_1) =
- \frac{\mu_1\mu_2}{\mu_1 + \mu_2}\mS^{(s)} \frac{\partial \jump{u}^{(-)}}{\partial x_1}, \quad x_1 < 0,
\end{equation}
\begin{equation}
\av{\sigma}^{(+)}(x_1) =
- \frac{\mu_1\mu_2}{\mu_1 + \mu_2}\mS^{(c)} \frac{\partial \jump{u}^{(-)}}{\partial x_1}, \quad x_1 > 0,
\end{equation}
where $\mS^{(s)}$ is a singular integral with the Cauchy type kernel and $\mS^{(c)}$ is a fixed-point singular operator
\begin{equation}
\mS^{(s)}\varphi(x_1) = \frac{1}{\pi} \int_{-\infty}^{0} \frac{\varphi(t)}{x_1 - t} dt, \quad x_1 < 0,
\end{equation}
\begin{equation}
\mS^{(c)}\varphi(x_1) = \frac{1}{\pi} \int_{-\infty}^{0} \frac{\varphi(t)}{x_1 - t} dt, \quad x_1 > 0.
\end{equation}

\paragraph{Mode I and II.} Using the results of Theorems \ref{teo1} and \ref{teo2} it is easy to show that in case of perfect interface, $\|\bK\| \to 0$, the integral identities for plane strain deformation reduce to
\begin{equation}
\label{general_2D}
\av{\bp} + \bmA^{(s)} \jump{\bp} = \bmB^{(s)} \frac{\partial \jump{\bu}^{(-)}}{\partial x_1}, \quad x_1 < 0,
\end{equation}
\begin{equation}
\label{general_2D2}
\av{\bsigma_2}^{(+)} + \bmA^{(c)} \jump{\bp} = \bmB^{(c)} \frac{\partial \jump{\bu}^{(-)}}{\partial x_1}, \quad x_1 > 0,
\end{equation}
where $\bmA^{(s)}$, $\bmB^{(s)}: F(\Reals_-) \to F(\Reals_-)$, and $\bmA^{(c)}, \bmB^{(c)}: F(\Reals_-) \to F(\Reals_+)$ are the following matrix operators
\begin{equation}
\label{op1}
\bmA^{(s)} =
\frac{b}{2(b^2 - d^2)}
\left[
(b\alpha - d\gamma) \bI + (d\alpha - b\gamma) \bE \mS^{(s)}
\right], \quad
\bmB^{(s)} =
-\frac{1}{b^2 - d^2}
\left[
b \bI \mS^{(s)} - d \bE
\right].
\end{equation}
\begin{equation}
\label{op2}
\bmA^{(c)} =
\frac{b(d\alpha - b\gamma)}{2(b^2 - d^2)} \bE \mS^{(c)}, \quad
\bmB^{(c)} =
-\frac{b}{b^2 - d^2} \bI \mS^{(c)}.
\end{equation}

\end{document}